\documentclass[11pt]{article}

\usepackage[backend=biber,style=alphabetic,maxbibnames=100]{biblatex}
\usepackage{fullpage}
\usepackage[utf8]{inputenc}
\usepackage[dvipsnames,svgnames,x11names,hyperref]{xcolor}
\usepackage[colorlinks]{hyperref}
\usepackage{enumerate}
\usepackage{nicefrac}
\usepackage{graphicx}
\usepackage[textsize=scriptsize]{todonotes}
\usepackage{algorithm}
\usepackage{amsmath}
\usepackage{mathtools}
\usepackage{authblk}
\usepackage{amssymb}
\usepackage{amsthm}
\usepackage{esvect}
\usepackage[normalem]{ulem}
\usepackage{braket}
\usepackage{pgfplots}
\usepackage{multicol}

\pgfplotsset{compat=1.6}

\bibliography{bibliography}

\hypersetup{
	colorlinks = true,
	citecolor = teal,
	urlcolor = [rgb]{0.6,0.2,0.2},
}

\newtheorem{theorem}{Theorem}
\newtheorem{lemma}[theorem]{Lemma}

\DeclareMathOperator{\poly}{poly}

\DeclareMathOperator{\OPT}{OPT}
\DeclareMathOperator{\SMC}{SMC}

\let\originalleft\left
\let\originalright\right
\renewcommand{\left}{\mathopen{}\mathclose\bgroup\originalleft}
\renewcommand{\right}{\aftergroup\egroup\originalright}

\DeclareFontFamily{U}{mathx}{\hyphenchar\font45}
\DeclareFontShape{U}{mathx}{m}{n}{
      <5> <6> <7> <8> <9> <10>
      <10.95> <12> <14.4> <17.28> <20.74> <24.88>
      mathx10
      }{}
\DeclareSymbolFont{mathx}{U}{mathx}{m}{n}
\DeclareMathSymbol{\bigtimes}{1}{mathx}{"91}

\makeatletter
\newcommand\dotline{\@ifnextchar[
  \answerlinetowidth\answerlinetoeol}
\newcommand\answerlinetowidth[1][0pt]{\hbox to #1{\leaders\hbox to \answerdotsep{\hss.\hss}\hfill}}
\newcommand\answerlinetoeol{\leaders\hbox to \answerdotsep{\hss.\hss}\hfill\strut}
\newcommand\answerdotsep{0.187cm}
\makeatother

\makeatletter
\newenvironment{breakablealgorithm}
  {
     \refstepcounter{algorithm}
     \hrule height.8pt depth0pt \kern2pt
     \renewcommand{\caption}[2][\relax]{
       {\raggedright\textbf{\fname@algorithm~\thealgorithm} ##2\par}%
       \ifx\relax##1\relax 
         \addcontentsline{loa}{algorithm}{\protect\numberline{\thealgorithm}##2}%
       \else 
         \addcontentsline{loa}{algorithm}{\protect\numberline{\thealgorithm}##1}%
       \fi
       \kern2pt\hrule\kern2pt
     }
  }{
     \kern2pt\hrule\relax
  }
\makeatother

\DeclareMathOperator{\e}{e}

\title{Quantum speedups for dynamic programming \\ on  $n$-dimensional lattice graphs}
\author[1]{Adam Glos}
\author[2]{Martins Kokainis}
\author[3]{Ryuhei Mori}
\author[2]{Jevgēnijs Vihrovs}

\affil[1]{Institute of Theoretical and Applied Informatics,\authorcr Polish Academy of Sciences, ul.~Bałtycka 44-100 Gliwice, Poland}
\affil[2]{Centre for Quantum Computer Science, Faculty of Computing,\authorcr University of Latvia, Raiņa 19, Riga, Latvia, LV-1586}
\affil[3]{School of Computing, Tokyo Institute of Technology,\authorcr 2-12-1, Ookayama, Meguro-ku, 152-8550, Japan}

\date{}

\begin{document}

\maketitle

\begin{abstract}
    Motivated by the quantum speedup for dynamic programming on the Boolean hypercube by Ambainis et al.~(2019), we investigate which graphs admit a similar quantum advantage.
    In this paper, we examine a generalization of the Boolean hypercube graph, the $n$-dimensional lattice graph $Q(D,n)$ with vertices in $\{0,1,\ldots,D\}^n$.
    We study the complexity of the following problem: given a subgraph $G$ of $Q(D,n)$ via query access to the edges, determine whether there is a path from $0^n$ to $D^n$.
    While the classical query complexity is $\widetilde{\Theta}((D+1)^n)$, we show a quantum algorithm with complexity $\widetilde O(T_D^n)$, where $T_D < D+1$.
    The first few values of $T_D$ are $T_1 \approx 1.817$, $T_2 \approx 2.660$, $T_3 \approx 3.529$, $T_4 \approx 4.421$, $T_5 \approx 5.332$.
    We also prove that $T_D \geq \frac{D+1}{\e}$, thus for general $D$, this algorithm does not provide, for example, a speedup, polynomial in the size of the lattice.
    
    While the presented quantum algorithm is a natural generalization of the known
    quantum algorithm for $D=1$ by Ambainis et al., the analysis of complexity is rather complicated.
    For the precise analysis, we use the saddle-point method, which is a
    common tool in analytic combinatorics, but has not been widely used
    in this field.
    
    We then show an implementation of this algorithm with time complexity $\poly(n)^{\log n} T_D^n$, and apply it to the \textsc{Set Multicover} problem.
    In this problem, $m$ subsets of $[n]$ are given, and the task is to find the smallest number of these subsets that cover each element of $[n]$ at least $D$ times.
    While the time complexity of the best known classical algorithm is $O(m(D+1)^n)$, the time complexity of our quantum algorithm is $\poly(m,n)^{\log n} T_D^n$.
\end{abstract}

\section{Introduction}

Dynamic programming (DP) algorithms have been widely used to solve various NP-hard problems in exponential time.
Bellman, Held and Karp showed how DP can be used to solve the \textsc{Travelling Salesman Problem} in  $\widetilde{O}(2^n)$\footnote{$f(n)=\widetilde{O}(g(n))$ if $f(n) = O(\log^c(g(n)) g(n))$ for some constant $c$.} time using DP \cite{Bel62,HK62}, which still remains the most efficient classical algorithm for this problem.
Their technique can be used to solve a plethora of different problems \cite{FK10,Bodlaender2012}.

The DP approach of Bellman, Held and Karp solves the subproblems corresponding to subsets of an $n$-element set, sequentially in increasing order of the subset size.
This typically results in an $\widetilde{\Theta}(2^n)$ time algorithm, as there are $2^n$ distinct subsets.
What kind of speedups can we obtain for such algorithms using quantum computers?

It is natural to consider applying Grover's search, which is known to speed up some algorithms for NP-complete problems.
For example, we can use it to search through the $2^n$ possible assignments to the SAT problem instance on $n$ variables in $\widetilde{O}(\sqrt{2^n})$ time.
However, it is not immediately clear how to apply it to the DP algorithm described above.
Recently, Ambainis et al.~showed a quantum algorithm that combines classical precalculation with recursive applications of Grover's search that solves such DP problems in $\widetilde{O}(1.817^n)$ time, assuming the QRAM model of computation \cite{ABIKPV19}.

In their work, they examined the transition graph of such a DP algorithm, which can be seen as a directed $n$-dimensional Boolean hypercube, with edges connecting smaller weight vertices to larger weight vertices.
A natural question arises, for what other graphs there exist quantum algorithms that achieve a speedup over the classical DP?
In this work, we examine a generalization of the hypercube graph, the $n$-dimensional lattice graph with vertices in $\{0,1,\ldots,D\}^n$.

While the classical DP for this graph has running time $\widetilde{\Theta}((D+1)^n)$, as it examines all vertices, we prove that there exists a quantum algorithm (in the QRAM model) that solves this problem in time $\poly(n)^{\log n} T_D^n$ for $T_D < D+1$ (Theorems \ref{thm:main}, \ref{thm:time}).
Our algorithm essentially is a generalization of the algorithm of Ambainis et al.
We show the following running time for small values of $D$:
\begin{table}[H] \label{tbl:intro}
\begin{center}
\begin{tabular}{c||c|c|c|c|c|c}
        $D$ & $1$     & $2$     & $3$     & $4$     & $5$     & $6$ \\ \hline\hline
        $T_D$ & $1.81692$ & $2.65908$ & $3.52836$ & $4.42064$ & $5.33149$ & $6.25720$
\end{tabular}
\end{center}
\caption{The complexity of the quantum algorithm.}
\end{table}
A detailed summary of our numerical results is given in Section \ref{sec:smalld}.
Note that the case $D=1$ corresponds to the hypercube, where we have the same complexity as Ambainis et al.
In our proofs, we extensively use the saddle point method from analytic combinatorics to estimate the asymptotic value of the combinatorial expressions arising from the complexity analysis.

We also prove a lower bound on the query complexity of the algorithm for general $D$.
Our motivation is to check whether our algorithm, for example, could achieve complexity $\widetilde O((D+1)^{cn})$ for large $D$ for some $c < 1$.
We prove that this is not the case: more specifically, for any $D$, the algorithm performs at least $\widetilde\Omega\left(\left(\frac{D+1}{\e}\right)^n\right)$ queries (Theorem \ref{thm:lb}).

As an example application, we apply our algorithm to the \textsc{Set Multicover} problem (SMC), which is a generalization of the \textsc{Set Cover} problem.
In this problem, the input consists of $m$ subsets of the $n$-element set, and the task is to calculate the smallest number of these subsets that together cover each element at least $D$ times, possibly with overlap and repetition.
While the best known classical algorithm has running time $O(m(D+1)^n)$ \cite{Nederlof08, HWYL10}, our quantum algorithm has running time $\poly(m,n)^{\log n} T_D^n$, improving the exponential complexity (Theorem \ref{thm:smc}).

The paper is organized as follows.
In Section \ref{sec:prelim}, we formally introduce the $n$-dimensional lattice graph and some of the notation used in the paper.
In Section \ref{sec:problem}, we define the generic query problem that models the examined DP.
In Section \ref{sec:algo}, we describe our quantum algorithm.
In Section \ref{sec:query}, we establish the query complexity of this algorithm and prove the aforementioned lower bound.
In Section \ref{sec:time}, we discuss the implementation of this algorithm and establish its time complexity.
Finally, in Section \ref{sec:app}, we show how to apply our algorithm to SMC, and discuss other related problems.

\section{Preliminaries} \label{sec:prelim}

The $n$-dimensional lattice graph is defined as follows.
The vertex set is given by $\{0,1,\ldots,D\}^n$, and the edge set consists of directed pairs of two vertices $u$ and $v$ such that $v_i = u_i+1$ for exactly one $i$, and $u_j = v_j$ for $j \neq i$.
We denote this graph by $Q(D,n)$.
Alternatively, this graph can be seen as the Cartesian product of $n$ paths on $D+1$ vertices.
The case $D=1$ is known as the Boolean hypercube and is usually denoted by $Q_n$.

We define the \emph{weight} of a vertex $x \in V$ as the sum of its coordinates $|x| \coloneqq \sum_{i=1}^n x_i$.
Denote $x \leq y$ iff for all $i \in [n]$, $x_i \leq y_i$ holds. If additionally $x \neq y$, denote such relation by $x < y$.

Throughout the paper we use the standard notation $[n] \coloneqq \{1,\ldots,n\}$.
In Section \ref{sec:smc}, we use notation for the superset $2^{[n]} \coloneqq \{S \mid S \subseteq [n]\}$ and for the characteristic vector $\chi(S) \in \{0,1\}^n$ of a set $S \in [n]$ defined as $\chi(S)_i = 1$ iff $i \in S$, and $0$ otherwise.

We write $f(n) = \poly(n)$ to denote that $f(n) = O(n^c)$ for some constant $c$.
We also write $f(n,m) = \poly(n,m)$ to denote that $f(n,m) = O(n^c m^d)$ for some constants $c$ and $d$.

For a multivariable polynomial $p(x_1,\ldots,x_m)$, we denote by $[x_1^{c_1}\cdots x_m^{c_m}] p(x_1,\ldots,x_m)$ its coefficient at the multinomial $x_1^{c_1}\cdots x_m^{c_m}$.

\section{Path in the hyperlattice} \label{sec:problem}

We formulate our generic problem as follows.
The input to the problem is a subgraph $G$ of $Q(D,n)$.
The problem is to determine whether there is a path from $0^n$ to $D^n$ in $G$.
We examine this as a query problem: a single query determines whether an edge $(u,v)$ is present in $G$ or not.

Classically, we can solve this problem using a dynamic programming algorithm that computes the value $\text{dp}(v)$ recursively for all $v$, which is defined as $1$ if there is a path from $0^n$ to $v$, and $0$ otherwise.
It is calculated by the Bellman, Held and Karp style recurrence \cite{Bel62, HK62}:
$$\text{dp}(v) = \bigvee_{(u,v) \in E}\{ \text{dp}(u) \land ((u,v) \in G)\}, \hspace{1cm} \text{dp}(0^n) = 1.$$
The query complexity of this algorithm is $O(n (D+1)^n)$.
From this moment we refer to this as the \emph{classical dynamic programming algorithm}.

The query complexity is also lower bounded by $\widetilde{\Omega}((D+1)^n)$.
Consider the sets of edges $E_W$ connecting the vertices with weights $W$ and $W+1$, $$E_W \coloneqq \{(u,v) \mid (u,v) \in Q(D,n), |u| = W, |v| = W+1\}.$$
Since the total number of edges is equal to $(D+1)^{n-1} Dn$, there is such a $W$ that $|E_W| \geq (D+1)^{n-1} Dn/Dn = (D+1)^{n-1}$ (in fact, one can prove that the largest size is achieved for $W=\lfloor nD/2 \rfloor$ \cite{dBvETK51}, but it is not necessary for this argument).
Any such $E_W$ is a cut of $H_D$, hence any path from $0^n$ to $D^n$ passes through $E_W$.
Examine all $G$ that contain exactly one edge from $E_W$, and all other edges.
Also examine the graph that contains no edges from $E_W$, and all other edges.
In the first case, any such graph contains a desired path, and in the second case there is no such path.
To distinguish these cases, one must solve the OR problem on $|E_W|$ variables.
Classically, $\Omega(|E_W|)$ queries are needed (see, for example, \cite{BdW02}).
Hence, the classical (deterministic and randomized) query complexity of this problem is $\widetilde{\Theta}((D+1)^n)$.
This also implies $\widetilde{\Omega}(\sqrt{(D+1)^n})$ quantum lower bound for this problem \cite{BBBV97}.

\section{The quantum algorithm} \label{sec:algo}

Our algorithm closely follows the ideas of \cite{ABIKPV19}.
We will use the well-known generalization of Grover's search:
\begin{theorem}[Variable time quantum search, Theorem 3 in \cite{Amb10}] \label{thm:vts}
Let $\mathcal A_1$, $\ldots$, $\mathcal A_N$ be quantum algorithms that compute a function $f : [N] \to \{0,1\}$ and have query complexities $t_1$, $\ldots$, $t_N$, respectively, which are known beforehand.
Suppose that for each $\mathcal A_i$, if $f(i) = 0$, then $A_i = 0$ with certainty, and if $f(i) = 1$, then $A_i = 1$ with constant success probability.
Then there exists a quantum algorithm with constant success probability that checks whether $f(i) = 1$ for at least one $i$ and has query complexity
$O\left(\sqrt{t_1^2+\ldots+t_N^2}\right).$
Moreover, if $f(i) = 0$ for all $i\in[N]$, then the algorithm outputs $0$ with certainty.
\end{theorem}
Even though Ambainis formulates the main theorem for zero-error inputs, the statement above follows from the construction of the algorithm.

Now we describe our algorithm.
We solve a more general problem: suppose $s, t \in \{0,1,\ldots,D\}^n$ are such that $s < t$ and we are given a subgraph of the $n$-dimensional lattice with vertices in
$$\bigtimes_{i=1}^n \{s_i,\ldots,t_i\},$$
and the task is to determine whether there is path from $s$ to $t$.
We need this generalized problem because our algorithm is recursive and is called for sublattices.

Define $d_i \coloneqq t_i-s_i$.
Let $n_d$ be the number of indices $i \in [n]$ such that $d_i = d$.
Note that the minimum and maximum weights of the vertices of this lattice are $|s|$ and $|t|$, respectively.

We call a set of vertices with fixed total weight a \emph{layer}.
The algorithm will operate with $K$ layers (numbered $1$ to $K$), with the $k$-th having weight $|s|+W_k$, where $W_{k} \coloneqq \left\lfloor \sum_{d=1}^D \alpha_{k,d} d n_d\right\rfloor$.
Denote the set of vertices in this layer by

$$\mathcal L_k \coloneqq \left\{v \mid |v| = |s|+W_{k}\right\}.$$

Here, $\alpha_{k,d} \in (0,1/2)$ are constant parameters that have to be determined before we run the algorithm.
The choice of $\alpha_{k,d}$ does not depend on the input to the algorithm, similarly as it was in \cite{ABIKPV19}.
For each $k \in [K]$ and $d \in [D]$, we require that $\alpha_{k,d} < \alpha_{k+1,d}$.
In addition to the $K$ layers defined in this way, we also consider the $(K+1)$-th layer $\mathcal L_{K+1}$, which is the set of vertices with weight $|s|+W_{K+1}$, where $W_{K+1}\coloneqq \left\lfloor \frac{|t|-|s|}{2} \right\rfloor$.
We can see that the weights $W_1, \ldots, W_{K+1}$ defined in this way are non-decreasing.

\bigskip

\begin{breakablealgorithm} \label{alg:main}
\caption{The quantum algorithm for detecting a path in the hyperlattice.}

\textsc{Path($s$, $t$):}
\begin{enumerate}
\item \label{itm:sc0} 
Calculate $n_1$, $\ldots$, $n_D$, and $W_1$, $\ldots$, $W_{K+1}$.
If $W_{k}=W_{k+1}$ for some $k$, determine whether there exists a path from $s$ to $t$ using classical dynamic programming and return.

\item \label{itm:sc1} Otherwise, first perform the precalculation step.
Let $\text{dp}(v)$ be $1$ iff there is a path from $s$ to $v$.
Calculate $\text{dp}(v)$ for all vertices $v$ such that $|v| \leq |s|+W_1$ using classical dynamic programming.
Store the values of $\text{dp}(v)$ for all vertices with $|v| = |s|+W_1$.

Let $\text{dp}'(v)$ be $1$ iff there is a path from $v$ to $t$.
Symmetrically, we also calculate $\text{dp}'(v)$ for all vertices with $|v| = |t| - W_1$.

\item \label{itm:sc2} Define the function $\textsc{LayerPath}(k,v)$ to be $1$ iff there is a path from $s$ to $v$ such that $v \in \mathcal L_k$.
Implement this function recursively as follows.
\begin{itemize}
    \item $\textsc{LayerPath}(1,v)$ is read out from the stored values.
    \item For $k > 1$, run VTS over the vertices $u \in \mathcal L_{k-1}$ such that $u < v$.
    The required value is equal to
    $$\textsc{LayerPath}(k,v) = \bigvee_u \left\{ \textsc{LayerPath}(k-1,u) \land \textsc{Path}(u,v) \right\}.$$
\end{itemize}

\item \label{itm:sc3} Similarly define and implement the function $\textsc{LayerPath}'(k,v)$, which denotes the existence of a path from $v$ to $t$ such that $v \in \mathcal L_k'$ (where $\mathcal L_k'$ is the layer with weight $|t|-W_k$).
To find the final answer, run VTS over the vertices in the middle layer $v \in \mathcal L_{K+1}$ and calculate
$$\bigvee_v\left\{ \textsc{LayerPath}(K+1,v) \land \textsc{LayerPath}'(K+1,v)\right\}.$$
\end{enumerate}
\end{breakablealgorithm}

\section{Query complexity} \label{sec:query}

For simplicity, let us examine the lattice
$$\bigtimes_{i=1}^n \{0,\ldots,t_i-s_i\},$$
as the analysis is identical.
Let the number of positions with maximum coordinate value $d$ be $n_d$.
We make an ansatz that the exponential complexity can be expressed as
$$T(n_1,\ldots,n_D) \coloneqq T_1^{n_1} T_2^{n_2} \cdot \ldots \cdot T_D^{n_D}$$
for some values $T_1, T_2, \ldots, T_D > 1$ (we also can include $n_0$ and $T_0$, however, $T_0 = 1$ always and doesn't affect the complexity).
We prove it by constructing generating polynomials for the precalculation and quantum search steps, and then approximating the required coefficients asymptotically.
We use the saddle point method that is frequently used for such estimation, specifically the theorems developed in \cite{BM04}.

\subsection{Generating polynomials} \label{sec:poly}
First we estimate the number of edges of the hyperlattice queried in the precalculation step.
The algorithm queries edges incoming to the vertices of weight at most $W_1$, and each vertex can have at most $n$ incoming edges.
The size of any layer with weight less than $W_1$ is at most the size of the layer with weight exactly $W_1$, as the size of the layers is non-decreasing until weight $W_{K+1}$ \cite{dBvETK51}.
Therefore, the number of queries during the precalculation is at most $n \cdot W_1 \cdot |\mathcal L_1| \leq n^2 D |\mathcal L_1|$, as $W_1 \leq nD$.
Since we are interested in the exponential complexity, we can omit $n$ and $D$, thus the exponential query complexity of the precalculation is given by  $|\mathcal L_1|$.

Now let $P_d(x) \coloneqq \sum_{i=0}^d x^i$.
The number of vertices of weight $W_1$ can be written as the coefficient at $x^{W_1}$ of the generating polynomial
$$P(x) \coloneqq \prod_{d=0}^D P_d(x)^{n_d}.$$
Indeed, each $P_d(x)$ in the product corresponds to a single position $i \in [n]$ with maximum value $d$ and the power of $x$ in that factor represents the coordinate of the vertex in this position.
Therefore, the total power that $x$ is raised to is equal to the total weight of the vertex, and coefficient at $x^{W_1}$ is equal to the number of vertices with weight $W_1$.
Since the total query complexity of the algorithm is lower bounded by this coefficient, we have
\begin{equation} \label{eq:prec}
T(n_1,\ldots,n_D) \geq \left[x^{W_1}\right]P(x).
\end{equation}

Similarly, we construct polynomials for the \textsc{LayerPath} calls.
Consider the total complexity of calling \textsc{LayerPath} recursively until some level $1 \leq k \leq K$ and then calling \textsc{Path} for a sublattice between levels $\mathcal L_k$ and $\mathcal L_{k+1}$.
Define the variables for the vertices chosen by the algorithm at level $i$ (where $k \leq i \leq K+1$) by $v^{(i)}$.
The \textsc{Path} call is performed on a sublattice between vertices $v^{(k)}$ and $v^{(k+1)}$, see Fig.~\ref{fig:rhombus}. 

\begin{figure}[h]
    \centering
    \includegraphics{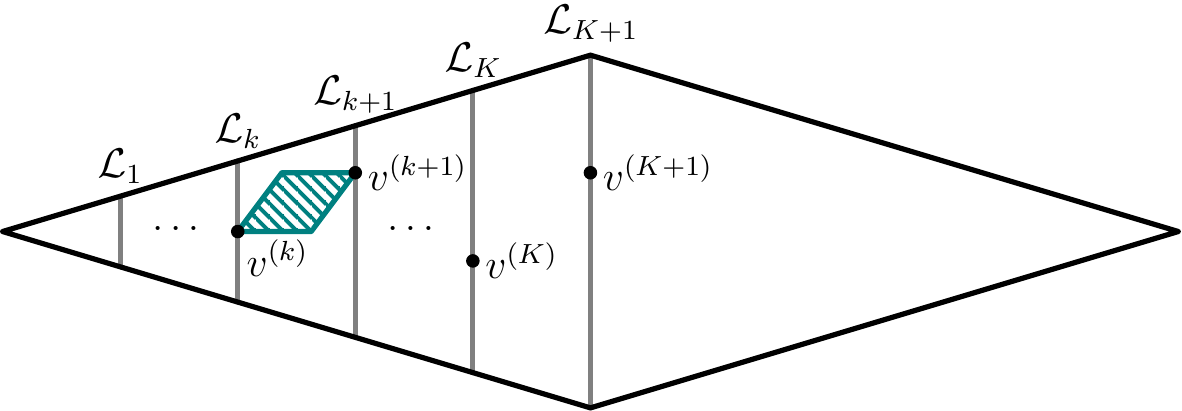}
    \caption{The choice of the vertices $v^{(i)}$ and the application of \textsc{Path} on the sublattice.}
    \label{fig:rhombus}
\end{figure}

Define
$$S_{k,d}(x_{k,k},\ldots,x_{k,K+1}) \coloneqq \sum_{i=0}^d T_i^2 \cdot \sum_{\substack{p_k,\ldots,p_{K+1}\in[0,d]\\ p_{k+1} \leq \ldots \leq p_{K+1} \\ p_{k+1}-p_k=i}}{}\prod_{j=k}^{K+1} x_{k,j}^{p_j}.$$
Again, this corresponds to a single coordinate.
The variable $x_{k,j}$ corresponds to the vertex $v^{(j)}$ and the power $p_j$ corresponds to the value of $v^{(j)}$ in that coordinate.

Examine the following multivariate polynomial:
$$S_k(x_{k,k},\ldots,x_{k,K+1}) \coloneqq \prod_{d=0}^D S_{k,d}^{n_d}(x_{k,k},\ldots,x_{k,K+1}).$$
We claim that the coefficient
$$\left[x_{k,k}^{W_k} \cdots x_{k,K+1}^{W_{K+1}}\right]S_k(x_{k,k},\ldots,x_{k,K+1})$$
is the required total complexity squared.

First of all, note that the value of this coefficient is the sum of $t^2$, where $t$ is the variable for the running time of \textsc{Path} between $v^{(k)}$ and $v^{(k+1)}$, for all choices of vertices $v^{(k)}$, $v^{(k+1)}$, $\ldots$, $v^{(K+1)}$.
Indeed, the powers $p_j$ encode the values of coordinates of $v^{(j)}$, and a factor of $T_i^2$ is present for each multinomial that has $p_{k+1}-p_k=i$ (that is, $v^{(k+1)}_l-v^{(k)}_l=i$ for the corresponding position $l$).

Then, we need to show that the sum of $t^2$ equals the examined running time squared.
Note that the choice of each vertex $v^{(j)}$ is performed using VTS.
In general, if we perform VTS on the algorithms with running times $s_1$, $\ldots$, $s_N$, then the total squared running time is equal to $s_1^2+\ldots+s_N^2$ by Theorem \ref{thm:vts}.
By repeating this argument in our case inductively at the choice of each vertex $v^{(j)}$, we obtain that the final squared running time indeed is the sum of all $t^2$.

Therefore, the square of the total running time of the algorithm is lower bounded by
\begin{equation} \label{eq:src}
T(n_1,\ldots,n_D)^2 \geq \left[x_{k,k}^{W_k} \cdots x_{k,K+1}^{W_{K+1}}\right]S_k(x_{k,k},\ldots,x_{k,K+1}).
\end{equation}

Together the inequalities (\ref{eq:prec}) and (\ref{eq:src}) allow us to estimate $T$.
The total time complexity of the quantum algorithm is twice the sum of the coefficients given in Eq.~(\ref{eq:prec}) and (\ref{eq:src}) for all $k \in [K]$ (twice because of the calls to $\textsc{LayerPath}$ and its symmetric counterpart $\textsc{LayerPath}'$).
This is upper bounded by $2K$ times the maximum of these coefficients.
Since $2K$ is a constant, and there are $O(\log n)$ levels of recursion (see the next section), in total this contributes only $(2K)^{O(\log n)} = \poly(n)$ factor to the total complexity of the quantum algorithm.

\subsubsection{Depth of recursion}
Note that the algorithm stops the recursive calls if for at least one $k$, we have $W_k = W_{k+1}$, in which case it runs the classical dynamic programming on the whole sublattice at step \ref{itm:sc0}.
That happens when
$$\left\lfloor \sum_{d=1}^D \alpha_{k,d} d n_d\right\rfloor = \left\lfloor \sum_{d=1}^D \alpha_{k+1,d} d n_d\right\rfloor.$$
If this is true, then we also have
$$\sum_{d=1}^D \alpha_{k+1,d} d n_d - \sum_{d=1}^D \alpha_{k,d} d n_d = c$$
for some constant $c < 1$.
By regrouping the terms, we get
$$\sum_{d=1}^D (\alpha_{k+1,d} - \alpha_{k,d}) d n_d = c.$$
Denote $h \coloneqq \min_{d \in [D]} \{ \alpha_{k+1,d} - \alpha_{k,d} \}$.
Then
$$\sum_{d=1}^D dn_d \leq \frac{c}{h}.$$
Note that the left hand side is the maximum total weight of a vertex.
However, at each recursive call the difference between the vertices with the minimum and maximum total weights decreases twice, since the VTS call at step \ref{itm:sc3} runs over the vertices with weight half the current difference.
Since $c$ and $h$ is constant, after $O(\log(nD)) = O(\log n)$ recursive calls the recursion stops.
Moreover, the classical dynamic programming then runs on a sublattice of constant size, hence adds only a factor of $O(1)$ to the overall complexity.

Lastly, we can address the contribution of the constant factor of VTS from Theorem \ref{thm:vts} to the complexity of our algorithm.
At one level of recursion there are $K+1$ nested applications of VTS, and there are $O(\log n)$ levels of recursion.
Therefore, the total overhead incurred is $O(1)^{O(K \log n)} = \poly(n)$, since $K$ is a constant.



\subsection{Saddle point approximation}

In this section, we show how to describe the tight asymptotic complexity of $T(n_1,\ldots,n_D)$ using the saddle point method (a detailed review can be found in \cite{FS09}, Chapter VIII).
Our main technical tool will be the following theorem.
\begin{theorem} \label{thm:asy}
Let $p_1(x_1,\ldots,x_m)$, $\ldots$, $p_D(x_1,\ldots,x_m)$ be polynomials with non-negative coefficients.
Let $n$ be a positive integer and $b_1, \ldots, b_D$ be non-negative rational numbers such that $b_1 + \ldots + b_D = 1$ and $b_d n$ is an integer for all $d \in [D]$.
Let $a_{i,d}$ be rational numbers (for $i \in [m]$, $d \in [D]$) and $\alpha_i \coloneqq a_{i,1}b_1+\ldots+a_{i,D}b_D$.
Suppose that $\alpha_i n$ are integer for all $i \in [m]$.
Then
\begin{enumerate}[(1)]
    \item \label{itm:upp} $\left[ x_1^{\alpha_1 n}\cdots x_m^{\alpha_m n}\right]\prod_{d=1}^D p_d(x_1,\ldots,x_m)^{b_d n} \leq \left(\inf_{x_1,\ldots,x_m>0} \prod_{d=1}^D \left(\frac{p_d(x_1,\ldots,x_m)}{x_1^{a_{1,d}}\cdots x_m^{a_{m,d}}}\right)^{b_d}\right)^n$
    \item \label{itm:low} $\left[ x_1^{\alpha_1 n}\cdots x_m^{\alpha_m n}\right]\prod_{d=1}^D p_d(x_1,\ldots,x_m)^{b_d n} = \Omega\left(\left(\inf_{x_1,\ldots,x_m>0} \prod_{d=1}^D \left(\frac{p_d(x_1,\ldots,x_m)}{x_1^{a_{1,d}}\cdots x_m^{a_{m,d}}}\right)^{b_d}\right)^n \right)$,
where $\Omega$ depends on the variable $n$.
\end{enumerate}
\end{theorem}

\begin{proof}
To prove this, we use the following saddle point approximation.\footnote{Setting $\gamma = 1$ in the statement of the original theorem.}
\begin{theorem}[Saddle point method, Theorem 2 in \cite{BM04}] \label{thm:saddle}
Let $p(x_1,\ldots,x_m)$ be a polynomial with non-negative coefficients.
Let $\alpha_1, \ldots, \alpha_m$ be some rational numbers and let $n_i$ be the series of all integers $j$ such that $\alpha_k j$ are integers and $\left[x_1^{\alpha_1 j}\cdots x_m^{\alpha_m j}\right] p(x_1,\ldots,x_m)^j \neq 0$.
Then
$$\lim_{i \to \infty} \frac{1}{n_i} \log \left( \left[x_1^{\alpha_1 n_i}\cdot \ldots\cdot x_m^{\alpha_m n_i}\right] p(x_1,\ldots,x_m)^{n_i} \right)= \inf_{x_1,\ldots,x_m > 0} \log \left( \frac{p(x_1,\ldots,x_m)}{x_1^{\alpha_1} \cdot \ldots \cdot x_m^{\alpha_m}} \right).$$
\end{theorem}
Let $p(x_1,\ldots,x_m) \coloneqq \prod_{d=1}^D p_d(x_1,\ldots,x_m)^{b_d}$, then
\begin{align*}
    \frac{p(x_1,\ldots,x_m)}{x_1^{\alpha_1} \cdots x_m^{\alpha_m}} 
    = \frac{\prod_{d=1}^D p_d(x_1,\ldots,x_m)^{b_d}}{x_1^{\alpha_1} \cdots x_m^{\alpha_m}}
    = \prod_{d=1}^D \frac{p_d(x_1,\ldots,x_m)^{b_d}}{x_1^{a_{1,d}b_d}\cdots x_m^{a_{m,d}b_d}} 
    = \prod_{d=1}^D \left( \frac{p_d(x_1,\ldots,x_m)}{x_1^{a_{1,d}}\cdots x_m^{a_{m,d}}} \right)^{b_d}.
\end{align*}
For the first part, as $p(x_1,\ldots,x_m)^n$ has non-negative coefficients, the coefficient at the multinomial $x_1^{\alpha_1 n} \cdots x_m^{\alpha_m n}$ is upper bounded by $\inf_{x_1,\ldots,x_m>0} \frac{p(x_1,\ldots,x_m)^n}{x_1^{\alpha_1 n} \cdots x_m^{\alpha_m n}} = \left(\inf_{x_1,\ldots,x_m>0} \frac{p(x_1,\ldots,x_m)}{x_1^{\alpha_1} \cdots x_m^{\alpha_m}}\right)^n$.
The second part follows directly by Theorem \ref{thm:saddle}.
\end{proof}

\subsubsection{Optimization program} To determine the complexity of the algorithm, we construct the following optimization problem.
Recall that the Algorithm \ref{alg:main} is given by the number of layers $K$ and the constants $\alpha_{k,d}$ that determine the weight of the layers, so assume they are fixed known numbers.
Assume that $\alpha_{k,d}$ are all rational numbers between $0$ and $1/2$ for $k \in [K]$; indeed, we can approximate any real number with arbitrary precision by a rational number.
Also let $T_0 = 1$ and $\alpha_{K+1,d} = 1/2$ for all $d \in [D]$ for convenience.

Examine the following program $\text{OPT}(D,K,\{\alpha_{k,d}\})$:
\begin{align*}
\text{minimize } T_D     \hspace{1cm}\text{s.t.}\hspace{1cm}    & T_d \geq \frac{P_d(x)}{x^{\alpha_{1,d}d}} & \forall d \in [D]\\
                                                                & T_d^2 \geq \frac{S_{k,d}(x_{k,k},\ldots,x_{k,K+1})}{x_{k,k}^{\alpha_{k,d}d}\cdots x_{k,K+1}^{\alpha_{K+1,d}d}} & \forall d \in [D], \forall k\in [K]\\
                                                                & T_d \geq 1 & \forall d \in [D]\\
                                                                & x > 0 \\
                                                                & x_{k,j} > 0 & \forall k \in [K], \forall j \in \{k,\ldots,K+1\}
\end{align*}

Let $n\coloneqq n_1+\ldots+n_D$ and $\alpha_k \coloneqq \frac{\sum_{d=1}^D \alpha_{k,d} d n_d}{n}$.
Suppose that $T_1, \ldots, T_D$ is a feasible point of the program. 
Then by Theorem \ref{thm:asy} (\ref{itm:upp}) (setting $b_i \coloneqq n_i/n$ and $a_{i,d} \coloneqq \alpha_{i,d} d$) we have
$$[x^{\alpha_1 n}] P(x) \leq \inf_{x>0} \prod_{d=1}^D \left(\frac{P_d(x)}{x^{\alpha_{1,d}d}}\right)^{n_d} \leq T_1^{n_1}\cdots T_D^{n_D}.$$
Similarly,
\begin{align*}
[x_{k,k}^{\alpha_k n}\cdots x_{k,K+1}^{\alpha_{K+1} n}] S_k(x_{k,k},\ldots,x_{k,K+1})
&\leq \inf_{x_{k,k},\ldots,x_{k,K+1}>0} \prod_{d=1}^D \left(\frac{S_{k,d}(x_{k,k},\ldots,x_{k,K+1})}{x_{k,k}^{\alpha_{k,d}d}\cdots x_{k,K+1}^{\alpha_{K+1,d}d}}\right)^{n_d} \\
&\leq (T_1^{n_1}\cdots T_D^{n_D})^2.
\end{align*}
Therefore, the program provides an upper bound on the complexity.

There are two subtleties that we need to address for correctness.
\begin{itemize}
    \item The numbers $\alpha_k n$ might not be integer; in Algorithm \ref{alg:main}, the weights of the layers are defined by $W_k = \lfloor \alpha_k n\rfloor$.
    This is a problem, since the inequalities in the program use precisely the numbers $\alpha_{k,d}$.
    Examine the coefficient $[x_1^{\lfloor\alpha_1 n\rfloor}\cdots x_m^{\lfloor\alpha_m n\rfloor}]p(x_1,\ldots,x_m)$ in such general case (when we need to round the powers).
    Let $\delta_k \coloneqq \alpha_k n - \lfloor \alpha_k n \rfloor$, here $0 \leq \delta_k < 1$.
    Then, by Theorem \ref{thm:asy} (\ref{itm:upp}),
    \begin{align*}
        \left[x_1^{\lfloor\alpha_1 n\rfloor}\cdots x_m^{\lfloor\alpha_m n\rfloor}\right]p(x_1,\ldots,x_m)^n 
        &\leq \inf_{x_1,\ldots,x_m \geq 0} \frac{p(x_1,\ldots,x_m)^n}{x_1^{\alpha_1 n - \delta_1}\cdots x_m^{\alpha_m n - \delta_m}} = (*)
    \end{align*}
    Now let $\hat x_1,\ldots,\hat x_m$ be the arguments that achieve $\inf_{x_1,\ldots,x_m \geq 0} \frac{p(x_1,\ldots,x_m)}{x_1^{\alpha_1}\cdots x_m^{\alpha_m}}$.
    Since $0 \leq \delta_k < 1$, we have $\hat x_k^{\delta_k} \leq \max\{\hat x_k,1\}$.
    Hence
    \begin{align*}
        (*) &\leq (\hat x_1^{\delta_1}\cdots \hat x_m^{\delta_m})\cdot\frac{p(\hat x_1,\ldots,\hat x_m)^n}{\hat x_1^{\alpha_1 n}\cdots \hat x_m^{\alpha_m n}} 
        \leq \left(\prod_{k=1}^m \max\{\hat x_k,1\}\right) \cdot \left(\inf_{x_1,\ldots,x_m \geq 0} \frac{p(x_1,\ldots,x_m)}{x_1^{\alpha_1}\cdots x_m^{\alpha_m}}\right)^n.
    \end{align*}
    As the additional factor is a constant, we can ignore it in the complexity.
    \item The second issue is when $W_k=W_{k+1}$ for some $k$. Then according to Algorithm \ref{alg:main}, we run the classical algorithm with complexity $\widetilde{\Theta}((D+1)^n)$.
    However, in that case $n$ is constant (see Section \ref{sec:poly}, Depth of recursion), which gives only a constant factor to the complexity.
\end{itemize}

\subsubsection{Optimality of the program}

In the start of the analysis, we made an assumption that the exponential complexity $T(n_1,\ldots,n_D)$ can be expressed as $T_1^{n_1}\cdots T_D^{n_D}$.
Here we show that the optimization program (which gives an upper bound on the complexity) can indeed achieve such value and gives the best possible solution.
\begin{itemize}
    \item First, we prove that $\text{OPT}(D,K,\{\alpha_{k,d}\})$ has a feasible solution.
For that, we need to show that all polynomials in the program can be upper bounded by a constant for some fixed values of the variables.

    First of all, $\frac{P_d(x)}{x^{\alpha_{1,d}d}}$ is upper bounded by $d+1$ (setting $x=1$).
    Now fix $k$ and examine the values $\frac{S_{k,d}(x_{k,k},\ldots,x_{k,K+1})}{x_{k,k}^{\alpha_{k,d}d}\cdots x_{k,K+1}^{\alpha_{K+1,d}d}}$.
    Examine only such assignments of the variables $x_{k,j}$ that $x_{k,k}x_{k,k+1}=1$ and $x_{k,j}=1$ for all other $j > k+1$.
    Now we write the polynomial as a univariate polynomial $S_{k,d}(y) \coloneqq S_{k,d}(1/y,y,1,1,\ldots,1)$.
    Note that for any summand of $S_{k,d}(y)$, if it contains some $T_i^2$ as a factor, then it is of the form $x_{k,k}^{p_k}x_{k,k+1}^{p_k+i}\cdot T_i^2 = y^i T_i^2$.
    Hence the polynomial can be written as $S_{k,d}(y) = \sum_{i=0}^d c_i y^i T_i^2$ for some constants $c_1, \ldots, c_d$.
    From this we can rewrite the corresponding program inequality and express $T_d^2$:
    \begin{align}
        T_d^2 &\geq \frac{\sum_{i=0}^d c_i y^i T_i^2}{y^{(\alpha_{k+1,d}-\alpha_{k,d})d}} \label{eq:qs}\\
        T_d^2 &\geq \frac{\sum_{i=0}^{d-1} c_i y^i T_i^2}{y^{(\alpha_{k+1,d}-\alpha_{k,d})d}} + y^{(1-\alpha_{k+1,d}+\alpha_{k,d})d}c_d T_d^2 \nonumber\\
        T_d^2 &\geq \frac{1}{1-y^{(1-\alpha_{k+1,d}+\alpha_{k,d})d}c_d} \cdot \frac{\sum_{i=0}^{d-1} c_i y^i T_i^2}{y^{(\alpha_{k+1,d}-\alpha_{k,d})d}}. \nonumber
    \end{align}
    Note that $c_d$ are constants that do not depend on $T_i$.
    If the right hand side is negative, then it follows that the original inequality Eq.~(\ref{eq:qs}) does not hold.
    Thus we need to pick such $y$ that the right hand side is positive for all $d$.
    Hence we require that
    $$y < \left(\frac{1}{c_d}\right)^{\frac{1}{(1-\alpha_{k+1,d}+\alpha_{k,d})d}}.$$
    Since the right hand side is a constant that does not depend on $T_i$, we can pick such $y$ that satisfies this inequality for all $d$.
    Then it follows that all $T_i$ is also upper bounded by some constants (by induction on $i$).
    
    \item Now the question remains whether the optimal solution to $\text{OPT}(D,K,\{\alpha_{k,d}\})$ gives the optimal complexity.
    That is, is the complexity $T_1^n\cdots T_D^{n_D}$ given by the optimal solution of the optimization program such that $T_D$ is the smallest possible?
    
    Suppose that indeed the complexity of the algorithm is upper bounded by $T_1^n\cdots T_D^{n_D}$ for some $T_1$, $\ldots$, $T_D$.
    We will derive a corresponding feasible point for the optimization program.
    
    Examine the complexity of the algorithm for $n_1 = b_1 n, \ldots, n_D = b_D n$ for some fixed rational $b_i$ such that $b_1+\ldots+b_D=1$.
    The coefficients of the polynomials $P$ and $S_k$ give the complexity of the corresponding part of the algorithm (precalculation, and quantum search until the $k$-th level, respectively).
    Such coefficients are of the form $\left[ x_1^{\alpha_1 n}\cdots x_m^{\alpha_m n}\right]\prod_{d=1}^D p_d(x_1,\ldots,x_m)^{n_d}$.
    Let $A_d \coloneqq T_d$, if $p = P$, and $A_d \coloneqq T_d^2$, if $p = S_k$.
    Then we have
    $$A_1^{n_1}\cdots A_D^{n_D} \geq \left[ x_1^{\alpha_1 n}\cdots x_m^{\alpha_m n}\right]\prod_{d=1}^D p_d(x_1,\ldots,x_m)^{n_d} = (*)$$
    On the other hand,
    $$(*) = \Omega\left(\left(\inf_{x_1,\ldots,x_m>0} \prod_{d=1}^D \left(\frac{p_d(x_1,\ldots,x_m)}{x_1^{a_{1,d}}\cdots x_m^{a_{m,d}}}\right)^{b_d}\right)^n \right)$$
    when $n$ grows large by Theorem \ref{thm:asy} (\ref{itm:low}) (setting $a_{i,d} \coloneqq \alpha_{i,d} d$).
    Then, in the limit $n \to \infty$, we have
    \begin{equation} \label{eq:lb}
        A_1^{b_1}\cdots A_D^{b_D} \geq \inf_{x_1,\ldots,x_m>0} \prod_{d=1}^D \left(\frac{p_d(x_1,\ldots,x_m)}{x_1^{a_{1,d}}\cdots x_m^{a_{m,d}}}\right)^{b_d}.
    \end{equation}
    
    Now let $\Delta_{D-1}$ be the standard $D$-simplex defined by $\{ b \in \mathbb R^D \mid b_1+\ldots+b_D = 1, b_d \geq 0\}$.
    Define $F_d(x) \coloneqq \frac{p_d(x_1,\ldots,x_m)}{x_1^{a_{1,d}}\cdots x_m^{a_{m,d}}}$, and $F(b,x) \coloneqq \prod_{d=1}^D F_d(x)^{b_d}$ for $b \in \Delta_{D-1}$ and $x \in \mathbb R_{> 0}^m$.
    
    First, we prove that that for a fixed $b$, the function $F(b,x)$ is strictly convex.
    Examine the polynomial $p_d(x_1,\ldots,x_m)$, which is either $P_d(x)$ or $S_{k,d}(x_{k,k},\ldots,x_{k,K+1})$.
    It was shown in \cite{Good57}, Theorem 6.3 that if the coefficients of $p_d(x_1,\ldots,x_m)$ are non-negative, and the points $(c_1,\ldots,c_m)$, at which $$\left[x_1^{c_1}\cdots x_m^{c_m}\right]p_d(x_1,\ldots,x_m) > 0,$$
    linearly span an $m$-dimensional space, then $\log(F_d(x))$ is a strictly convex function.
    If $p_d = P_d$, then this property immediately follows, because there is just one variable $x$ and the polynomial is non-constant.
    For $p_d = S_{k,d}$, the polynomial consists of summands of the form $T_{c_{k+1}-c_{k}}^2 x_{k,k}^{c_k}x_{k,k+1}^{c_{k+1}}\cdots x_{k,K+1}^{c_{K+1}}$, for $c_k \leq c_{k+1} \leq \ldots \leq c_{K+1}$.
    Note that the coefficient $T_{c_{k+1}-c_{k}}^2$ is positive.
    Thus the points $(c_k,\ldots,c_{K+1}) = (0,\ldots,0,1,\ldots,1)$ indeed linearly span a $(K-k+2)$-dimensional space.
    Therefore, $\log(F_d(x))$ is strictly convex.
    Then also the function $\sum_{d=1}^D b_d \log(F_d(x)) = \log(F(b,x))$ is strictly convex (for fixed $b$), as the sum of strictly convex functions is convex.
    Therefore, $F(b,x)$ is strictly convex as well.
    
    Therefore, the argument $\hat x(b)$ achieving $\inf_{x \in \mathbb R_{> 0}^m} F(b,x)$ is unique.
    Let $\hat F_d(b) \coloneqq F_d(\hat x(b))$ and define $D$ subsets of the simplex $C_d := \{b \in \Delta_{D-1} \mid \hat F_d(b) \leq A_d \}$.
    We will apply the following result for these sets:
    \begin{theorem}[Knaster-Kuratowski-Mazurkiewicz lemma \cite{KKM29}]
    Let the vertices of $\Delta_{D-1}$ be labeled by integers from $1$ to $D$.
    Let $C_1$, $\ldots$, $C_D$ be a family of closed sets such that for any $I \subseteq [D]$, the convex hull of the vertices labeled by $I$ is covered by $\cup_{d \in I} C_d$.
    Then $\cap_{d \in [D]} C_d \neq \varnothing$.
    \end{theorem}
    
    We check that the conditions of the lemma apply to our sets.
    First, note that $F(b,x)$ is continuous and strictly convex for a fixed $b$, hence $\hat x(b)$ is continuous and thus $\hat F_d(b)$ is continuous as well.
    Therefore, the ``threshold'' sets $C_d$ are closed.
    
    Secondly, let $I \subseteq [D]$ and examine a point $b$ in the convex hull of the simplex vertices labeled by $I$.
    For such a point, we have $b_d = 0$ for all $d \not\in I$.
    For the indices $d \in I$, for at least one we should have $\hat F_d(b) \leq A_d$, otherwise the inequality in Eq.~(\ref{eq:lb}) would be contradicted.
    Note that it was stated only for rational $b$, but since $\hat F_d(b)$ are continuous and any real number can be approximated with a rational number to arbitrary precision, the inequality also holds for real $b$.
    Thus indeed any such $b$ is covered by $\cup_{d \in I} C_d$.
    
    Therefore, we can apply the lemma and it follows that there exists a point $b \in \Delta_{D-1}$ such that $A_d \geq \hat F_d(b)$ for all $d \in [D]$.
    The corresponding point $\hat x(b)$ is a feasible point for the examined set of inequalities in the optimization program.
\end{itemize}

\subsubsection{Total complexity}
Finally, we will argue that there exists such a choice for $\{\alpha_{k,d}\}$ that $$\OPT(D,K,\{\alpha_{k,d}\}) < D+1.$$
Examine the algorithm with only $K=1$; the optimal complexity for any $K > 1$ cannot be larger, as we can simulate $K$ levels with $K+1$ levels by setting $\alpha_{2,d}=\alpha_{1,d}+\epsilon$ for $\epsilon \to 0$ for all $d \in [D]$.
For simplicity, denote $\alpha_d := \alpha_{1,d}$.

\begin{itemize}
\item Now examine the precalculation inequalities in $\OPT(D,1,\{\alpha_{1,d}\})$.
For any values of $\alpha_{1,d}$, if we set $x=1$, we have
$\frac{P_d(x)}{x^{\alpha_d d}} = \frac{\sum_{i=0}^d x^i} {x^{\alpha_d d}} = d+1$.
The derivative is equal to
$$\left( \frac{\sum_{i=0}^d x^i} {x^{\alpha_d d}} \right)' = \frac{x^{\alpha_d d} \cdot \sum_{i=1}^d i x^{i-1} - \alpha_d d x^{\alpha_d d - 1} \cdot \sum_{i=0}^d x^i}{x^{2\alpha_d d}} = \frac{d(d+1)}{2}-\alpha_d d(d+1)$$
at point $x = 1$.
Thus when $\alpha_d < \frac{1}{2}$, the derivative is positive.
It means that for arbitrary $\alpha_d < \frac{1}{2}$, there exists some $x(d)$ such that $\frac{P_d(x)}{x^{\alpha_d d}} < d+1$, and $\frac{P_d(x)}{x^{\alpha_d d}}$ monotonically grows on $x \in [x(d),1]$.
Thus, for arbitrary setting of $\{\alpha_d\}$ such that $\alpha_d < \frac{1}{2}$ for all $d \in [D]$, we can take  $\hat x \coloneqq \max_{d \in [D]} \{ x(d) \}$ as the common parameter, in which case all $\frac{P_d(\hat x)}{\hat x^{\alpha_d d}} < d+1$.

\item Now examine the set of the quantum search inequalities.
Let $y \coloneqq x_{1,1}$ and $z \coloneqq x_{1,2}$ for simplicity.
Then such inequalities are given by
$$T_d^2 \geq S_{1,d}(y,z) = \frac{\sum_{i=0}^d T_i^2 \sum_{p=0}^{d-i} y^p z^{p+i}}{y^{\alpha_d d} z^{d/2}}.$$
Now restrict the variables to condition $yz = 1$.
In that case, the polynomial above simplifies to
$$S_{1,d}(z) \coloneqq \frac{\sum_{i=0}^d T_i^2 \sum_{p=0}^{d-i} z^i}{y^{\frac d 2 + d\left(\alpha_d-\frac{1}{2}\right)} z^{d/2}} = \left(\sum_{i=0}^d T_i^2 (d-i+1) z^i\right)\cdot z^{d\left(\alpha_d-\frac{1}{2}\right)}.$$

We now find such values of $z$ and $\alpha_1, \ldots, \alpha_D$ so that $S_{1,d}(z) < (d+1)^2$ for all $d \in [D]$, where $T_1, \ldots, T_D$ are any values such that $T_d \leq d+1$ for all $d \in [D]$.
Denote $\hat S_{1,d}(z)$ to be $S_{1,d}(z)$ with $T_d = d+1$ for all $d \in [D]$, then $\hat S_{1,d}(z) < (d+1)^2$ as well.
Now let $T_d$ be the maximum of $\frac{P_d(\hat x)}{\hat x^{\alpha_d d}}$ from the previous bullet and $\hat S_{1,d}(z)$.
Then, $T_d < d+1$, and we have both $T_d \geq \frac{P_d(\hat x)}{\hat x^{\alpha_d d}}$ and $T_d^2 \geq \hat S_{1,d}(z) \geq S_{1,d}(z)$, since $S_{1,d}(z)$ cannot become larger when $T_d$ decrease.

Now we show how to find such $z$ and $\alpha_1, \ldots, \alpha_D$.
Examine the sum in the polynomial $\hat S_{1,d}(z)$
$$\sum_{i=0}^d (i+1)^2 (d-i+1) z^i = (d+1) + \sum_{i=1}^d (i+1)^2 (d-i+1) z^i.$$
Examine the second part of the sum.
We can find a sufficiently small value of $z \in (0,1)$ such that this part is smaller than any value $\epsilon > 0$ for all $d \in [D]$.
Now, let $\alpha_d = \frac{1}{2}-\frac{c}{d}$ for some constant $c > 0$.
Then
$$z^{d\left(\alpha_d-\frac{1}{2}\right)} = z^{-c}$$
for all $d \in [D]$.
Thus, the total value of the sum now is at most $(d+1+\epsilon)z^{-c}$.
As $z^{-1} > 1$, take a sufficiently small value of $c$ so that this value is at most $(d+1)^2$.
\end{itemize}

Therefore, putting all together, we have the main result:
\begin{theorem} \label{thm:main}
There exists a bounded-error quantum algorithm that solves the path in the $n$-dimensional lattice problem using $\widetilde O(T_D^n)$ queries, where $T_D < D+1$.
The optimal value of $T_D$ can be found by optimizing $\OPT(D,K,\{\alpha_{k,d}\})$ over $K$ and $\{\alpha_{k,d}\}$.
\end{theorem}

\subsection{Complexity for small \texorpdfstring{$\boldsymbol{D}$}{D}} \label{sec:smalld}

To find the estimate on the complexity for small values of $D$ and $K$, we have optimized the value of $\OPT(D,K,\{\alpha_{k,d}\})$ using Mathematica (minimizing over the values of $\alpha_{k,d}$).
Table \ref{tbl:numerical} compiles the results obtained by the optimization.
In case of $D=1$, we recovered the complexity of the quantum algorithm from \cite{ABIKPV19} for the path in the hypercube problem, which is a special case of our algorithm.

\begin{table}[ht]
\begin{center}
\begin{tabular}{c||c|c|c|c|c|c}
        & $D=1$     & $D=2$     & $D=3$     & $D=4$     & $D=5$     & $D=6$ \\ \hline\hline
$K=1$   & $1.86793$ & $2.76625$ & $3.68995$ & $4.63206$ & $5.58735$ & $6.55223$  \\ \hline
$K=2$   & $1.82562$ & $2.67843$ & $3.55933$ & $4.46334$ & $5.38554$ & $6.32193$  \\ \hline
$K=3$   & $1.81819$ & $2.66198$ & $3.53322$ & $4.42759$ & $5.34059$ & $6.26840$  \\ \hline
$K=4$   & $1.81707$ & $2.65939$ & $3.52893$ & $4.42148$ & $5.33263$ & $6.25862$  \\ \hline
$K=5$   & $1.81692$ & $2.65908$ & $3.52836$ & $4.42064$ & $5.33149$ & $6.25720$
\end{tabular}
\end{center}
\caption{The complexity of the quantum algorithm for small values of $D$ and $K$.}
\label{tbl:numerical}
\end{table}

For $K=1$, we were able to estimate the complexity for up to $D=18$.
Figure \ref{fig:k1} shows the values of the difference between $D+1$ and $T_D$ for this range. 
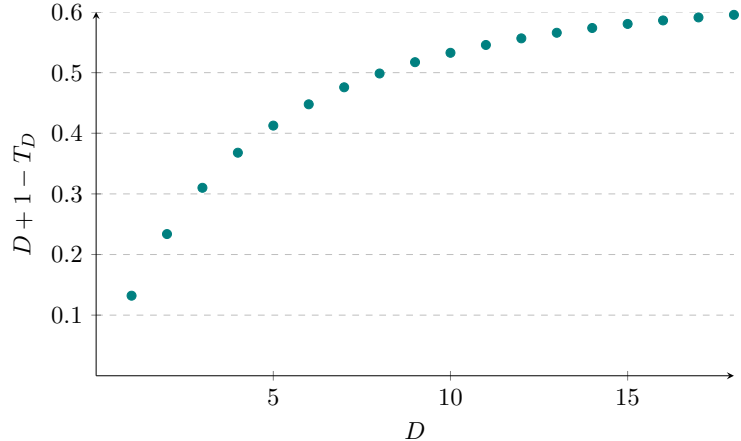
\begin{figure}[ht]
\centering
\begin{tikzpicture}[scale=0.85]
\begin{axis}[
    width=0.7\textwidth,
    height=\axisdefaultheight,
    axis lines = left,
    xlabel={$D$},
    ylabel={$D+1-T_D$},
    xmin=0, xmax=18,
    ymin=0, ymax=0.6,
    xtick={5,10,15},
    ytick={0.1,0.2,0.3,0.4,0.5,0.6},
    ymajorgrids=true,
    grid style=dashed,
]

\addplot[
    only marks,
    color=teal,
    mark size=2pt,
    ]
    coordinates {
    (1, 2-1.8679291102114184)
    (2, 3-2.7662504942190176)
    (3, 4-3.6899390889963155)
    (4, 5-4.632054318702819)
    (5, 6-5.5873596338697835)
    (6, 7-6.55222537443972)
    (7, 8-7.524152471011434)
    (8, 9-8.501400896330647)
    (9, 10-9.482736621759516)
    (10, 11-10.467266858165571)
    (11, 12-11.454333012714129)
    (12, 13-12.443440344113888)
    (13, 14-13.43421095829736)
    (14, 15-14.426351815325729)
    (15, 16-15.419632547557956)
    (16, 17-16.41386980382098)
    (17, 18-17.408916012469916)
    (18, 19-18.404651188768383)
    };
\end{axis}
\end{tikzpicture}
    
    \caption{The advantage of the quantum algorithm over the classical for $K=1$.}
    \label{fig:k1}
\end{figure}

Our Mathematica code used for determining the values of $T_D$ can be accessed at \url{https://doi.org/10.5281/zenodo.4603689}.
In Appendix \ref{app:numerical}, we list the parameters for the case $K=1$.

\subsection{Lower bound for general \texorpdfstring{$\boldsymbol{D}$}{D}} \label{sec:lb}

Even though Theorem \ref{thm:main} establishes the quantum advantage of the algorithm, it is interesting how large the speedup can get for large $D$.
In this section, we prove that the speedup cannot be substantial, more specifically:
\begin{theorem} \label{thm:lb}
For any fixed integers $D \geq 1$ and $K \geq 1$, Algorithm \ref{alg:main} performs $\widetilde\Omega\left(\left(\frac{D+1}{\e}\right)^n\right)$ queries on the lattice $Q(D,n)$.
\end{theorem}

\begin{proof}
The structure of the proof is as follows.
First, we prove that if $\alpha_{1,D} > \frac{1}{4}$, then the number of queries used in the algorithm during the precalculation step \ref{itm:sc1} is at least $\widetilde\Omega((0.664554(D+1))^n)$ queries (Lemma \ref{thm:plb} in Appendix \ref{app:lba}).
Then, we prove that if $\alpha_{1,D} \leq \frac{1}{4}$, then the quantum search part in steps \ref{itm:sc2} and \ref{itm:sc3} performs at least $\widetilde\Omega\left(\left(\frac{D+1}{\e}\right)^n\right)$ queries (Lemma \ref{thm:slb} in Appendix \ref{app:lba}).
Therefore, depending on whether $\alpha_{1,D} > \frac{1}{4}$, one of the precalculation or the quantum search performs $\widetilde\Omega((c(D+1))^n)$ queries for constant $c$, and the claim follows, since $\frac{1}{\e} < 0.664554$.
\end{proof}

\section{Time complexity} \label{sec:time}

In this section we examine a possible high-level implementation of the described algorithm and argue that there exists a quantum algorithm with the same exponential time complexity as the query complexity.

Firstly, we assume the commonly used QRAM model of computation that allows to access $N$ memory cells in superposition in time $O(\log N)$ \cite{GLM08}.
This is needed when the algorithm accesses the precalculated values of $\text{dp}$.
Since in our case $N$ is always at most $(D+1)^n$, this introduces only a $O(\log ((D+1)^n)) = O(n)$ additional factor to the time complexity.

The main problem that arises is the efficient implementation of VTS.
During the VTS execution, multiple quantum algorithms should be performed in superposition.
More formally, to apply VTS to algorithms $\mathcal A_1$, $\ldots$, $\mathcal A_N$, we should specify the \emph{algorithm oracle} that, given the index of the algorithm $i$ and the time step $t$, applies the $t$-th step of $\mathcal A_i$ (see Section 2.2 of \cite{CJOP20} for formal definition of such an oracle and related discussion).
If the algorithms $\mathcal A_i$ are unstructured, the implementation of such an oracle may take even $O(N)$ time (if, for example, all of the algorithms perform a different gate on different qubits at the $t$-th step).

We circumvent this issue by showing that it is possible to use only Grover's search to implement the algorithm, retaining the same exponential complexity (however, the sub-exponential factor in the complexity will increase).
Nonetheless, the use of VTS in the query algorithm not only achieves a smaller query complexity, but also allowed to prove the estimate on the exponential complexity, which would not be so amiable for the algorithm that uses Grover's search.

\subsection{Implementation}

The main idea of the implementation is to fix a ``class'' of vertices for each of the $2K+1$ layers examined by the algorithm, and do this for all $r = O(\log n)$ levels of recursion.
We will essentially define these classes by the number of coordinates of a vertex in such layer that are equal to $0$, $1$, $\ldots$, $D$.
Then, we can first fix a class for each layer for all levels of recursion classically.
We will show that there are at most $n^{D^2}$ different classes we have to consider at each layer.
Since there are $2K+1$ layers at one level of recursion, and $O(\log n)$ levels of recursion, this classical precalculation will take time $n^{O(D^2 K \log n)}$.
For each such choice of classes, we will run a quantum algorithm that checks for the path in the hyperlattice constrained on these classes of the vertices the path can go through.
The advantage of the quantum algorithm will come from checking the permutations of the coordinates using Grover's search.
The time complexity of the quantum part will be $n^{O(K \log n)} T_D^n$ ($T_d^n$ as in the query algorithm, and $n^{O(K \log n)}$ from the logarithmic factors in Grover's search), therefore the total time complexity will be $n^{O(D^2 K \log n)}\cdot n^{O(K \log n)} T_D^n=n^{O(D^2 K \log n)} T_D^n$, thus the exponential complexity stays the same.

\subsubsection{Layer classes}
In all of the applications of VTS in the algorithm, we use it in the following scenario: given a vertex $x$, examine all vertices $y$ with fixed weight $|y| = W$ such that $y < x$ (note that VTS over the middle layer $\mathcal L_{K+1}$ can be viewed in this way by taking $x$ to be the final vertex in the lattice, and VTS over the vertices in the layers symmetrical to $\mathcal L_{K+1}$ can be analyzed similarly).

We define a \emph{class} of $y$'s (in respect to $x$) in the following way.
Let $n_{a,b}$ be the number of $i \in [n]$ such that $y_i = a$ and $x_i = b$, where $a \leq b$.
All $y$ in the same class have the same values of $n_{a,b}$ for all $a$, $b$.
Also define a \emph{representative} of a class as a single particular $y$ from that class; we will define it as the lexicographically smallest such $y$.

As mentioned in the informal description above, we can fix the classes for all layers examined by the quantum algorithm and generate the corresponding representatives classically.
Note that in our quantum algorithm, recursive calls work with the sublattice constrained on the vertices $s \leq y \leq t$ for some $s < t$, so for each position of $y_i$ we should have also $y_i \geq s_i$; however, we can reduce it to lattice $0^n \leq y' \leq x$, where $x_i \coloneqq t_i-s_i$ for all $i$.
To get the real value of $y$, we generate a representative $y'$, and set $y_i \coloneqq y_i'+s_i$.

Consider an example for $D=2$.
The following figure illustrates the representative $y$ (note that the order of positions of $x$ here is lexicographical for simplicity, but it may be arbitrary).
\begin{figure}[H]
\centering
\begin{align*}
    x&=00\dotline[0.58cm]0\hspace{0.06cm}11\dotline[1.8cm]1\hspace{0.07cm}22\dotline[3cm]2\\
    y&=\underbrace{00\ldots0}_{n_{0,0}}\underbrace{00\ldots0}_{n_{0,1}}\underbrace{11\ldots1}_{n_{1,1}}\underbrace{00\ldots0}_{n_{0,2}}\underbrace{11\ldots1}_{n_{1,2}}\underbrace{22\ldots2}_{n_{2,2}}
\end{align*}
\caption{The (lexicographically smallest) representative for $y$ for $D=2$.}
\end{figure}

Note that $n_{a,b}$ can be at most $n$.
Therefore, there are at most $n^{D^2}$ choices for classes at each layer.
Thus the total number of different sets of choices for all layers is $n^{O(D^2 K \log n)}$.
For each such set of choices, we then run a quantum algorithm that checks for a path in the sublattice constrained on these classes.

\subsubsection{Quantum algorithm}

The algorithm basically implements Algorithm \ref{alg:main}, with VTS replaced by Grover's search.
Thus we only describe how we run the Grover's search.
We will also use the analysis of Grover's search with multiple marked elements.

\begin{theorem}[Grover's search]
Let $f : S \to \{0,1\}$, where $|S| = N$.
Suppose we can generate a uniform superposition $\frac{1}{\sqrt N} \sum_{x \in S} \ket{x}$ in $O(\poly(\log N))$ time, and there is a bounded-error quantum algorithm $\mathcal A$ that computes $f(x)$ with time complexity $T$.
Suppose also that there is a promise that either there are at least $k$ solutions to $f(x) = 1$, or there are none.
Then there exists a bounded-error quantum algorithm that runs in time $O(T \log N \sqrt{N/k})$, and detects whether there exists $x$ such that $f(x) = 1$.
\end{theorem}

\begin{proof}
First, it is well-known that in the case of $k$ marked elements, Grover's algorithm \cite{Grover96} needs $O(\sqrt{N/k})$ iterations.
Second, the gate complexity of one iteration of Grover's search is known to be $O(\log N)$.
Finally, even though $\mathcal A$ has constant probability of error, there is a result that implements Grover's search with a bounded-error oracle without introducing another logarithmic factor \cite{HMDw03}.
\end{proof}

Now, for a class $\mathcal C$ of $y$'s (for a fixed $x$) we need to generate a superposition $\frac{1}{\sqrt{|\mathcal C|}} \sum_{y \in \mathcal C} \ket{y}$ efficiently to apply Grover's algorithm.
We will generate a slightly different superposition for the same purposes.
Let $I_1, \ldots, I_D$ be sets $I_d \coloneqq \{i \in [n] \mid x_i = d\}$ and let $n_d \coloneqq |I_d|$.
Let $y_{\mathcal C}$ be the representative of $\mathcal C$.
We will generate the superposition
\begin{equation}
    \bigotimes_{d=0}^D \frac{1}{\sqrt{n_d!}} \sum_{\pi \in S_{n_d}} \ket{{\pi(y_{\mathcal C}}_{I_d})}\ket{\pi}, \label{eq:superposition}
\end{equation}
where ${y_{\mathcal C}}_{I_d}$ are the positions of $y_{\mathcal C}$ in $I_d$.

We need a couple of procedures to generate such state.
First, there exists a procedure to generate the uniform superposition of permutations $\frac{1}{\sqrt{n!}}\sum_{\pi \in S_n} \ket{\pi_1,\ldots,\pi_n}$ that requires $O(n^2 \log n)$ elementary gates \cite{AL97,CdLYMW19}.
Then, we can build a circuit with $O(\poly(n))$ gates that takes as an input $\pi \in S_n$, $s \in \{0,1,\ldots,D\}^n$ and returns $\pi(s)$.
Such an circuit essentially could work as follows: let $t \coloneqq 0^n$; then for each pair $i, j \in [n]$, check whether $\pi(i)=j$; if yes, let $t_j \leftarrow t_j+ s_{\pi(i)}$; in the end return $t$.
Using these two subroutines, we can generate the required superposition using $O(\poly(n))$ gates (we assume $D$ is a constant).

However, we do not necessarily know the sets $I_d$, because the positions of $x$ have been permuted by previous applications of permutations.
To mitigate this, note that we can access this permutation in its own register from the previous computation.
That is, suppose that $x$ belongs to a class $\mathcal C'$ and $x = \sigma(x_{\mathcal C '})$, where $x_{\mathcal C '}$ is the representative of $\mathcal C'$ generated by the classical algorithm from the previous subsection.
Then we have the state $\ket{\sigma(x_{\mathcal C '})}\ket{\sigma}$.

We can then apply $\sigma$ to both $\pi(y_{\mathcal C})$ and $\pi$.
That is, we implement the transformation
$$\ket{\pi(y_{\mathcal C})}\ket{\pi} \to \ket{\sigma(\pi(y_{\mathcal C}))}\ket{\sigma\pi}.$$
Such transformation can also be implemented in $O(\poly(n))$ gates.
Note that now we store the permutation $\sigma\pi$ in a separate register, which we use in a similar way recursively.

Finally, examine the number of positive solutions among $\pi(y_{\mathcal C})$.
That is, for how many $\pi$ there exists a path from $\pi(y)$ to $x$?
Suppose that there is a path from $y$ to $x$ for some $y \in \mathcal C$.
Examine the indices $I_d$; for $n_{a,d}$ of these indices $i$ we have $y_i = a$.
There are exactly $n_{a,d}!$ permutations that permute these indices and don't change $y$.
Hence, there are $\prod_{a = 0}^d n_{a,d}!$ distinct permutations $\pi \in S_{n_d}$ such that $\pi(y) = y$.

Therefore, there are $k\coloneqq\prod_{d=0}^D \prod_{a = 0}^d n_{a,d}!$ distinct permutations $\pi$ among the considered such that $\pi(y) = y$.
The total number of considered permutations is $N\coloneqq\prod_{d=0}^D n_d!$.
Among these permutations, either there are no positive solutions, or at least $k$ of the solutions are positive.
Grover's search then works in time $O(T \log N \sqrt{N/k})$.
In this case, $N/k$ is exactly the size of the class $\mathcal C$, because $\frac{n_d!}{n_{0,d}!\cdots n_{d,d}!}$ is the number of unique permutations of ${y_{\mathcal C}}_{P_d}$, the multinomial coefficient $\binom{n_d}{n_{0,d},\ldots,n_{d,d}}$.
Hence the state Eq.~(\ref{eq:superposition}) effectively replaces the need for the state $\frac{1}{\sqrt{|\mathcal C|}} \sum_{y \in \mathcal C} \ket{y}$.

\subsubsection{Total complexity}

Finally, we discuss the total time complexity of this algorithm.
The exponential time complexity of the described quantum algorithm is at most the exponential query complexity because Grover's search examines a single class $\mathcal C$, while VTS in the query algorithm examines all possible classes.
Since Grover's search has a logarithmic factor overhead, the total time complexity of the quantum part of the algorithm is what is described in Section~\ref{sec:query} multiplied by $n^{O(K \log n)}$, resulting in $n^{O(K \log n)} T_1^{n_1} \cdots T_D^{n_D}$.

Since there are $n^{O(D^2 K \log n)}$ sets of choices for the classes of the layers, the final total time complexity of the algorithm is $n^{O(D^2 K \log n)} T_1^{n_1} \cdots T_D^{n_D}$.
Therefore, we have the following result.
\begin{theorem} \label{thm:time}
Assuming QRAM model of computation, there exists a quantum algorithm that solves the path in the $n$-dimensional lattice problem and has time complexity $\poly(n)^{D^2 \log n} \cdot T_D^n$.
\end{theorem}

\section{Applications} \label{sec:app}

\subsection{Set multicover} \label{sec:smc}

As an example application of our algorithm, we apply it to the \textsc{Set Multicover} problem (SMC).
This is a generalization of the \textsc{Minimum Set Cover} problem.
The SMC problem is formulated as follows:\vspace{2mm}
    
\textbf{Input:} A set of subsets $\mathcal S \subseteq 2^{[n]}$, and a positive integer $D$.\vspace{2mm}
    
\textbf{Output:}  The size $k$ of the smallest tuple $(S_1,\ldots,S_k) \in \mathcal S^k$, such that for all $i\in \mathcal [n]$, we have $|\{j  \mid i \in S_j\}| \geq D$, that is, each element is covered at least $D$ times (note that each set $S \in \mathcal S$ can be used more than once).\vspace{2mm}

Denote this problem by $\SMC_D$, and $m \coloneqq |\mathcal S|$.
This problem has been studied classically, and there exists an exact deterministic algorithm based on the inclusion-exclusion principle that solves this problem in time $\widetilde O(m(D+1)^n)$ and polynomial space \cite{Nederlof08, HWYL10}.
While there are various approximation algorithms for this problem, we are not aware of a more efficient classical exact algorithm.

There is a different simple classical dynamic programming algorithm for this problem with the same time complexity (although it uses exponential space), which we can speed up using our quantum algorithm.
For a vector $x \in \{0,1,\ldots,D\}^n$, define $\text{dp}(x)$ to be the size $k$ of the smallest tuple $(\mathcal C_1,\ldots,\mathcal C_k) \in \mathcal S^k$ such that for each $i$, we have $|\{j \in [k] \mid i \in \mathcal C_j\}| \geq x_i$.
It can be calculated using the recurrence
\begin{align*}
    \text{dp}(0^n) = 0, \hspace{1cm}
    \text{dp}(x) = 1 + \min_{S \in \mathcal S} \{ \text{dp}(x') \},
\end{align*}
where $x'$ is given by $x_i' = \max\{0,x_i-\chi(S)_i\}$ for all $i$.
Consequently, the answer to the problem is equal to $\text{dp}(D^n)$.
The number of distinct $x$ is $(D+1)^n$, and $\text{dp}(x)$ for a single $x$ can be calculated in time $O(nm)$, if $\text{dp}(y)$ has been calculated for all $y < x$.
Thus the time complexity is $O(nm(D+1)^n)$ and space complexity is $O((D+1)^n)$.

Note that even though the state space of the dynamic programming here is $\{0,1,\ldots,D\}^n$, the underlying transition graph is not the same as the hyperlattice examined in the quantum algorithm.
A set $S \in \mathcal S$ can connect vertices that are $|S|$ distance apart from each other, unlike distance $1$ in the hyperlattice.
We can essentially reduce this to the hyperlattice-like transition graph by breaking such transition into $|S|$ distinct transitions.

More formally, examine pairs $(x,S)$, where $x \in \{0,1,\ldots,D\}^n$, $S \in \mathcal S$.
Let $e(x,S) := \min\{i \in S \mid x_i > 0\}$; if there is no such $i$, let $e(x,S)$ be $0$.
Define a new function
\begin{align*}
\text{dp}(x,S) &=
\begin{cases}
    0, &\text{if $x=0^n$,}\\
    \text{dp}(x-\chi(\{e(x,S)\}),S), &\text{if $e(x,S) > 0$,}\\
    1+\min_{T \in \mathcal S, e(x,T) > 0} \{ \text{dp}(x-\chi(\{e(x,T)\}),T\}, &\text{if $e(x,S) = 0$.}
\end{cases}
\end{align*}
The new recursion also solves $\SMC_D$, and the answer is equal to $\min_{S \in \mathcal S}\{ \text{dp}(D^n,S)\}$.

Examine the underlying transition graph between pairs $(x,S)$.
We can see that there is a transition between two pairs $(x,S)$ and $(y,T)$ only if $y_i = x_i+1$ for exactly one $i$, and $y_i = x_i$ for other $i$.
This is the $n$-dimensional lattice graph $Q(D,n)$.
Thus we can apply our quantum algorithm with a few modifications:
\begin{itemize}
    \item We now run Grover's search over $(x,S)$ with fixed $|x|$ for all $S \in \mathcal S$.
    This adds a $\poly(m,n)$ factor to each run of Grover's search.
    \item Since we are searching for the minimum value of $\text{dp}$, we actually need a quantum algorithm for finding the minimum instead of Grover's search.
    We can use the well-known quantum minimum finding algorithm that retains the same query complexity as Grover's search \cite{DH96}\footnote{Note that this algorithm assumes queries with zero error, but we apply it to bounded-error queries. However, it consists of multiple runs of Grover's search, so we can still use the result of \cite{HMDw03} to avoid the additional logarithmic factor.}.
    It introduces only an additional $O(\log n)$ factor for the queries of minimum finding to encode the values of $\text{dp}$, since $\text{dp}(x,S)$ can be as large as $Dn$.
    \item A single query for a transition between pairs $(x,S)$ and $(y,T)$ in this case returns the value of the value added to the dp at transition, which is either $0$ or $1$.
    If these pairs are not connected in the transition graph, the query can return $\infty$.
    Note that such query can be implemented in $\poly(m,n)$ time.
\end{itemize}

Since the total number of runs of Grover's search is $O(K \log n)$, the additional factor incurred is $\poly(m,n)^{O(K \log n)}$.
This provides a quantum algorithm for this problem with total time complexity $$\poly(m,n)^{O(K \log n)}\cdot n^{O(D^2 K \log n)} T_D^n = m^{O(K \log n)} n^{O(D^2 K \log n)} T_D^n.$$
Therefore, we have the following theorem.
\begin{theorem} \label{thm:smc}
Assuming the QRAM model of computation, there exists a quantum algorithm that solves $\SMC_D$ in time $\poly(m,n)^{\log n} T_D^n$, where $T_D < D+1$.
\end{theorem}

\subsection{Related problems}

We are aware of some other works that implement the dynamic programming on the $\{0,1,\ldots,D\}^n$ $n$-dimensional lattice.

Psaraftis examined the job scheduling problem \cite{Psaraftis80}, with application to aircraft landing scheduling.
The problem requires ordering $n$ groups of jobs with $D$ identical jobs in each group.
A cost transition function is given: the cost of processing a job belonging to group $j$ after processing a job belonging to group $i$ is given by $f(i,j,d_1,\ldots,d_n)$, where $d_i$ is the number of jobs left to process.
The task is to find an ordering of the $nD$ jobs that minimizes the total cost.
This is almost exactly the setting for our quantum algorithm, hence we get $\poly(n)^{\log n} T_D^n$ time quantum algorithm.
Psaraftis proposed a classical $O(n^2(D+1)^n)$ time dynamic programming algorithm.
Note that if $f(i,j,d_1,\ldots,d_n)$ are unstructured (can be arbitrary values), then there does not exist a faster classical algorithm by the lower bound of Section \ref{sec:problem}.

However, if $f(i,j,d_1,\ldots,d_n)$ are structured or can be computed efficiently by an oracle, there exist more efficient classical algorithms for these kinds of problems.
For instance, the many-visits travelling salesman problem (MV-TSP) asks for the shortest route in a weighted $n$-vertex graph that visits vertex $i$ exactly $D_i$ times.
In this case, $f(i,j,d_1,\ldots,d_n) = w(i,j)$, where $w(i,j)$ is the weight of the edge between $i$ and $j$.
The state-of-the-art classical algorithm by Kowalik et al.~solves this problem in $\widetilde O(4^n)$ time and space \cite{KLNSW20}.
Thus, our quantum algorithm does not provide an advantage.
It would be quite interesting to see if there exists a quantum speedup for this MV-TSP algorithm.

Lastly, Gromicho et al.~proposed an exact algorithm for the job-shop scheduling problem \cite{GvHST12,vHNOG17}.
In this problem, there are $n$ jobs to be processed on $D$ machines.
Each job consists of $D$ tasks, with each task to be performed on a separate machine.
The tasks for each job need to be processed in a specific order.
The time to process job $i$ on machine $j$ is given by $p_{ij}$.
Each machine can perform at most one task at any moment, but machines can perform the tasks in parallel.
The problem is to schedule the starting times for all tasks so as to minimize the last ending time of the tasks.
Gromicho et al.~give a dynamic programming algorithm that solves the problem in time $O((p_{\max})^{2n}(D+1))^n$, where $p_{\max} = \max_{i,j} \{ p_{ij} \}$.

The states of their dynamic programming are also vectors in $\{0,1,\ldots,D\}^n$: a state $x$ represents a partial completion of tasks, where $x_i$ tasks of job $i$ have already been completed.
Their dynamic programming calculates the set of task schedulings for $x$ that can be potentially extended to an optimal scheduling for all tasks.
However, it is not clear how to apply Grover's search to calculate a whole set of schedulings.
Therefore, even though the state space is the same as in our algorithm, we do not know whether it is possible to apply it in this case.

\section{Acknowledgements}
We would like to thank Krišjānis Prūsis for helpful discussions and comments.

A.G. has been supported in part by   National   Science   Center   under   grant   agreement 2019/32/T/ ST6/00158  and  2019/33/B/ST6/02011.  
M.K. has been supported by ``QuantERA ERA-NET Cofund in Quantum Technologies implemented within the European Union's Horizon 2020 Programme'' (QuantAlgo project).
R.M. was supported in part by JST PRESTO Grant Number JPMJPR1867
and JSPS KAKENHI Grant Numbers JP17K17711, JP18H04090, JP20H04138, and JP20H05966.
J.V. has been supported in part by the project ``Quantum algorithms: from complexity theory to experiment'' funded under ERDF programme 1.1.1.5.

\printbibliography

@article{dBvETK51,
    title = "On the set of divisors of a number",
    author = "{de Bruijn}, N. G. and {van Ebbenhorst Tengbergen}, $\text{C}^{\text{A}}$. and D. Kruyswijk",
    year = "1951",
    number = "2",
    volume = "23",
    pages = "191--193",
    journal = "Nieuw Archief voor Wiskunde",
    issn = "0028-9825",
    url = "https://research.tue.nl/en/publications/on-the-set-of-divisors-of-a-number"
}

@article{Good57,
    author = {I. J. Good},
    title = {{Saddle-point Methods for the Multinomial Distribution}},
    volume = {28},
    journal = {The Annals of Mathematical Statistics},
    number = {4},
    publisher = {Institute of Mathematical Statistics},
    pages = {861 -- 881},
    year = {1957},
    doi = {10.1214/aoms/1177706790}
}

@article{KKM29,
    author = {Bronisław Knaster and Casimir Kuratowski and Stefan Mazurkiewicz},
    journal = {Fundamenta Mathematicae},
    language = {ger},
    number = {1},
    pages = {132-137},
    title = {Ein Beweis des Fixpunktsatzes für $n$-dimensionale Simplexe},
    url = {http://eudml.org/doc/212127},
    volume = {14},
    year = {1929},
    doi = {10.4064/fm-14-1-132-137}
}

@article{GvHST12,
    title = {Solving the job-shop scheduling problem optimally by dynamic programming},
    journal = {Computers \& Operations Research},
    volume = {39},
    number = {12},
    pages = {2968-2977},
    year = {2012},
    issn = {0305-0548},
    doi = {https://doi.org/10.1016/j.cor.2012.02.024},
    author = {Joaquim A. S. Gromicho and Jelke J. {van Hoorn} and Francisco Saldanha-da-Gama and Gerrit T. Timmer}
}

@article{vHNOG17,
    title = {An corrigendum on the paper: Solving the job-shop scheduling problem optimally by dynamic programming},
    journal = {Computers \& Operations Research},
    volume = {78},
    pages = {381},
    year = {2017},
    issn = {0305-0548},
    doi = {https://doi.org/10.1016/j.cor.2016.09.001},
    author = {Jelke J. {van Hoorn} and Agustín Nogueira and Ignacio Ojea and Joaquim A. S. Gromicho}
}

@InProceedings{KLNSW20,
    author =	{Łukasz Kowalik and Shaohua Li and Wojciech Nadara and Marcin Smulewicz and Magnus Wahlstr{\"o}m},
    title =	{{Many Visits TSP Revisited}},
    booktitle =	{28th Annual European Symposium on Algorithms (ESA 2020)},
    pages =	{66:1--66:22},
    series =	{Leibniz International Proceedings in Informatics (LIPIcs)},
    ISBN =	{978-3-95977-162-7},
    ISSN =	{1868-8969},
    year =	{2020},
    volume =	{173},
    editor =	{Fabrizio Grandoni and Grzegorz Herman and Peter Sanders},
    publisher =	{Schloss Dagstuhl--Leibniz-Zentrum f{\"u}r Informatik},
    address =	{Dagstuhl, Germany},
    doi =		{10.4230/LIPIcs.ESA.2020.66},
	archivePrefix = {arXiv},
	eprint    = {2005.02329},
	primaryClass = "quant-ph"
}

@article{Psaraftis80,
    author = {Psaraftis, Harilaos N.},
    title = {A Dynamic Programming Approach for Sequencing Groups of Identical Jobs},
    journal = {Operations Research},
    volume = {28},
    number = {6},
    pages = {1347-1359},
    year = {1980},
    doi = {10.1287/opre.28.6.1347}
}

@article{HWYL10,
    title = {Dynamic programming based algorithms for set multicover and multiset multicover problems},
    journal = {Theoretical Computer Science},
    volume = {411},
    number = {26},
    pages = {2467-2474},
    year = {2010},
    issn = {0304-3975},
    doi = {https://doi.org/10.1016/j.tcs.2010.02.016},
    author = {Qiang-Sheng Hua and Yuexuan Wang and Dongxiao Yu and Francis C.M. Lau},
}

@mastersthesis{Nederlof08,
    title={Inclusion exclusion for hard problems},
    author={Jesper Nederlof},
    school={Utrecht University},
    location={Netherlands},
    year={2008},
    url={https://webspace.science.uu.nl/~neder003/MScThesis.pdf}
}

@article{AL97,
    title = {Simulation of Many-Body Fermi Systems on a Universal Quantum Computer},
    author = {Abrams, Daniel S. and Lloyd, Seth},
    journal = {Phys. Rev. Lett.},
    volume = {79},
    issue = {13},
    pages = {2586-2589},
    numpages = {0},
    year = {1997},
    month = {9},
    publisher = {American Physical Society},
    doi = {10.1103/PhysRevLett.79.2586},
	archivePrefix = {arXiv},
    eprint = {quant-ph/9703054}
}

@article{CdLYMW19,
    author = {Chiew, Mitchell and de Lacy, Kooper and Yu, Chao-Hua and Marsh, Sam and Wang, Jingbo B.},
    year = {2019},
    month = {08},
    pages = {302},
    title = {Graph comparison via nonlinear quantum search},
    volume = {18},
    journal = {Quantum Information Processing},
    doi = {10.1007/s11128-019-2407-2},
	archivePrefix = {arXiv},
	eprint    = {1810.01647},
	primaryClass = "quant-ph"
}

@inproceedings{HMDw03,
    author = {H\o{}yer, Peter and Mosca, Michele and de Wolf, Ronald},
    title = {Quantum Search on Bounded-Error Inputs},
    year = {2003},
    isbn = {3540404937},
    booktitle = {Automata, Languages and Programming},
    publisher = {Springer-Verlag},
    address = {Berlin, Heidelberg},
    pages = {291–299},
    numpages = {9},
    venue = {Eindhoven, The Netherlands},
    doi = {10.1007/3-540-45061-0_25},
    series = {ICALP'03},
	archivePrefix = {arXiv},
    eprint = {quant-ph/0304052}
}

@inproceedings{Grover96,
    author = {Grover, Lov K.},
    title = {A Fast Quantum Mechanical Algorithm for Database Search},
    year = {1996},
    isbn = {0897917855},
    publisher = {Association for Computing Machinery},
    address = {New York, NY, USA},
    doi = {10.1145/237814.237866},
    booktitle = {Proceedings of the Twenty-Eighth Annual ACM Symposium on Theory of Computing},
    pages = {212–219},
    numpages = {8},
    venue = {Philadelphia, Pennsylvania, USA},
    series = {STOC '96},
	archivePrefix = {arXiv},
    eprint = {quant-ph/9605043}
}

@book{FS09,
    author = {Flajolet, Philippe and Sedgewick, Robert},
    title = {Analytic Combinatorics},
    year = {2009},
    isbn = {0521898064},
    publisher = {Cambridge University Press},
    address = {USA},
    edition = {1},
    doi = {10.1017/CBO9780511801655}
}

@article{BM04,
    title = {Asymptotic enumeration methods for analyzing LDPC codes},
    author = {David Burshtein and Gadi Miller},
    journal = {IEEE Transactions on Information Theory},
    year = {2004},
    volume = {50},
    pages = {1115-1131},
    doi = {10.1109/TIT.2004.828064}
}

@article{BBBV97,
    author = {Bennett, Charles H. and Bernstein, Ethan and Brassard, Gilles and Vazirani, Umesh},
    title = {Strengths and Weaknesses of Quantum Computing},
    journal = {SIAM Journal on Computing},
    volume = {26},
    number = {5},
    pages = {1510-1523},
    year = {1997},
    doi = {10.1137/S0097539796300933},
	archivePrefix = {arXiv},
    eprint = {quant-ph/9701001},
}

@article{Amb10,
	author = {Ambainis, Andris},
	title = {Quantum Search with Variable Times},
	journal = {Theory of Computing Systems},
	year = {2010},
	volume = {47},
	number = {3},
	pages = {786--807},
	doi = {10.1007/s00224-009-9219-1},
	archivePrefix = {arXiv},
    eprint = {quant-ph/0609168}
}

@inproceedings{ABIKPV19,
    author = {Ambainis, Andris and Balodis, Kaspars and Iraids, Jānis and Kokainis, Martins and Prūsis, Krišjānis and Vihrovs, Jevgēnijs},
    title = {Quantum Speedups for Exponential-Time Dynamic Programming Algorithms},
    year = {2019},
    publisher = {Society for Industrial and Applied Mathematics},
    address = {USA},
    booktitle = {Proceedings of the Thirtieth Annual ACM-SIAM Symposium on Discrete Algorithms},
    pages = {1783–1793},
    numpages = {11},
    venue = {San Diego, California, USA},
    series = {SODA '19},
    doi = {10.1137/1.9781611975482.107},
	archivePrefix = {arXiv},
	eprint    = {1807.05209},
	primaryClass = "quant-ph"
}

@InProceedings{CJOP20,
    author =	{Arjan Cornelissen and Stacey Jeffery and Maris Ozols and Alvaro Piedrafita},
    title =	{{Span Programs and Quantum Time Complexity}},
    booktitle =	{45th International Symposium on Mathematical Foundations of Computer Science (MFCS 2020)},
    pages =	{26:1--26:14},
    series =	{Leibniz International Proceedings in Informatics (LIPIcs)},
    ISBN =	{978-3-95977-159-7},
    ISSN =	{1868-8969},
    year =	{2020},
    volume =	{170},
    editor =	{Javier Esparza and Daniel Kr{\'a}ľ},
    publisher =	{Schloss Dagstuhl--Leibniz-Zentrum f{\"u}r Informatik},
    venue = {Prague, Czechia},
    address =	{Dagstuhl, Germany},
    doi =		{10.4230/LIPIcs.MFCS.2020.26},
    eprint = {2005.01323},
	archivePrefix = {arXiv},
	primaryClass = "quant-ph"
}

@article{Bodlaender2012,
	author="Bodlaender, Hans L.
	and Fomin, Fedor V.
	and Koster, Arie M. C. A.
	and Kratsch, Dieter
	and Thilikos, Dimitrios M.",
	title="A Note on Exact Algorithms for Vertex Ordering Problems on Graphs",
	journal="Theory of Computing Systems",
	year="2012",
	volume="50",
	number="3",
	pages="420--432",
	doi="10.1007/s00224-011-9312-0"
}

@article{HK62,
	author = {Held, Michael and Karp, Richard M.},
	title = {A Dynamic Programming Approach to Sequencing Problems},
	journal = {Journal of SIAM},
	volume = {10},
	number = {1},
	pages = {196-210},
	year = {1962},
	doi = {10.1145/800029.808532}
}

@article{Bel62,
	author = {Bellman, Richard},
	title = {Dynamic Programming Treatment of the Travelling Salesman Problem},
	journal = {J. ACM},
	volume = {9},
	number = {1},
	year = {1962},
	pages = {61--63},
	publisher = {ACM},
	address = {New York, NY, USA},
	doi = {10.1145/321105.321111}
}

@misc{DH96,
	author = {Dürr, Christoph and Høyer, Peter},
	title = {A Quantum Algorithm for Finding the Minimum},
	archivePrefix = {arXiv},
    eprint = {quant-ph/9607014},
	year = {1996}
}

@article{BdW02,
	title = "Complexity measures and decision tree complexity: a survey",
	journal = "Theoretical Computer Science",
	volume = "288",
	number = "1",
	pages = "21-43",
	year = "2002",
	author = "Harry Buhrman and Ronald de Wolf",
	doi = {10.1016/S0304-3975(01)00144-X}
}

@book{FK10,
	title={Exact Exponential Algorithms},
	author={Fomin, Fedor V. and Kratsch, Dieter},
	year={2010},
	publisher={Springer Science \& Business Media},
	isbn={978-3-642-16533-7}
}

@article{GLM08,
	title = {Quantum Random Access Memory},
	author = {Giovannetti, Vittorio and Lloyd, Seth and Maccone, Lorenzo},
	journal = {Phys. Rev. Lett.},
	volume = {100},
	issue = {16},
	pages = {160501},
	numpages = {4},
	year = {2008},
	publisher = {American Physical Society},
    doi = {10.1103/PhysRevLett.100.160501},
	archivePrefix = {arXiv},
	eprint    = {0708.1879},
	primaryClass = "quant-ph"
}

\newpage

\appendix

\section{Numerical results for \texorpdfstring{$\boldsymbol{K=1}$}{K=1}} \label{app:numerical}

\begin{multicols}{3}

\paragraph*{\texorpdfstring{$\boldsymbol{D=1}$}{D=1}}

\begin{flalign*}
T_1 &= 1.86793\\
x &= 0.464808\\
x_{1,1} &= 6.0606\\
x_{1,2} &= 0.104715\\
\alpha_{1,1} &= 0.317317 &
\end{flalign*}

\vfill

\paragraph*{\texorpdfstring{$\boldsymbol{D=2}$}{D=2}}

\begin{flalign*}
T_1 &= 1.87788\\
T_2 &= 2.76626\\
x &= 0.595073\\
x_{1,1} &= 5.74769\\
x_{1,2} &= 0.12725\\
\alpha_{1,1} &= 0.314447\\
\alpha_{1,2} &= 0.337219&
\end{flalign*}

\vfill

\paragraph*{\texorpdfstring{$\boldsymbol{D=3}$}{D=3}}

\begin{flalign*}
T_1 &= 1.89454\\ 
T_2 &= 2.77944\\ 
T_3 &= 3.68995\\ 
x &= 0.684299\\ 
x_{1,1} &= 5.41613\\ 
x_{1,2} &= 0.146775\\ 
\alpha_{1,1} &= 0.310059\\ 
\alpha_{1,2} &= 0.336865\\ 
\alpha_{1,3} &= 0.351627&
\end{flalign*}

\vfill

\paragraph*{\texorpdfstring{$\boldsymbol{D=4}$}{D=4}}

\begin{flalign*}
T_1 &= 1.91039\\ 
T_2 &= 2.80346\\ 
T_3 &= 3.7035\\ 
T_4 &= 4.63207\\ 
x &= 0.747046\\ 
x_{1,1} &= 5.11625\\ 
x_{1,2} &= 0.163892\\ 
\alpha_{1,1} &= 0.306472\\ 
\alpha_{1,2} &= 0.335557\\ 
\alpha_{1,3} &= 0.351929\\ 
\alpha_{1,4} &= 0.362866&
\end{flalign*}

\vfill

\paragraph*{\texorpdfstring{$\boldsymbol{D=5}$}{D=5}}

\begin{flalign*}
T_1 &= 1.92386\\ 
T_2 &= 2.828\\ 
T_3 &= 3.72975\\ 
T_4 &= 4.64486\\ 
T_5 &= 5.58737\\ 
x &= 0.792588\\ 
x_{1,1} &= 4.8582\\ 
x_{1,2} &= 0.178964\\ 
\alpha_{1,1} &= 0.304026\\ 
\alpha_{1,2} &= 0.334429\\ 
\alpha_{1,3} &= 0.351624\\ 
\alpha_{1,4} &= 0.36331\\ 
\alpha_{1,5} &= 0.371992&
\end{flalign*}

\vfill

\paragraph*{\texorpdfstring{$\boldsymbol{D=6}$}{D=6}}

\begin{flalign*}
T_1 &= 1.93495\\ 
T_2 &= 2.85009\\ 
T_3 &= 3.75806\\ 
T_4 &= 4.6709\\ 
T_5 &= 5.600\\ 
T_6 &= 6.55224\\ 
x &= 0.826544\\ 
x_{1,1} &= 4.63595\\ 
x_{1,2} &= 0.192435\\ 
\alpha_{1,1} &= 0.302631\\ 
\alpha_{1,2} &= 0.333786\\ 
\alpha_{1,3} &= 0.351339\\ 
\alpha_{1,4} &= 0.363364\\ 
\alpha_{1,5} &= 0.372425\\ 
\alpha_{1,6} &= 0.379599&
\end{flalign*}

\vfill

\end{multicols}

\newpage

\section{Lower bound for general \texorpdfstring{$\boldsymbol{D}$}{D}} \label{app:lba}

	\subsection{Precalculation lower bound}
	
	\begin{lemma} \label{thm:plb}
	For every fixed $\alpha \in (0,0.5)$ there is a constant $c_\alpha$, depending only on $\alpha$, such that
	$$\inf_{x > 0} \frac{P_D(x)}{x^{\alpha D}} > (D+1) c_\alpha$$
	holds for all   $D \geq 1$. Furthermore, $\alpha > 1/4$ implies  $c_\alpha > c_{0.25} \approx 0.664554\ldots$.
		\end{lemma}
	\begin{proof}
		Recall that $ P_D(x)  = 1+x+\ldots + x^D $.
		By \cite[Theorem 1]{BM04},  for any $ \alpha \in (0,1) $ the infimum in question is attained at the unique positive solution of  
		$  \frac{x P'_D(x)}{P_D(x)} = D \alpha  $.
	The uniqueness of the solution implies that there is a mapping $ x_D : (0,1) \to  (0, +\infty) $,  where $ x_D(\alpha) $ is the  unique positive solution of 	$  \frac{x P'_D(x)}{P_D(x)} = D \alpha $. 
		
		Take into account that (for $ x\neq 1 $) $ P_D(x) = \frac{x^{D+1}-1}{x-1} $ and
		\begin{equation}\label{eq:l2e00}
			\frac{x P'_D(x)}{P_D(x)}
			=
			D + \frac{1}{1-x} + \frac{D+1}{x^{D+1} -1}.
		\end{equation}
		Define the value of the RHS of \eqref{eq:l2e00}  to be $ D/2 $ when $ x=1 $, then for each fixed $ D>0 $  the resulting function  is  continuous  and strictly increasing  on $ (0, +\infty) $, approaching values $ 0$ and $ D $ when $ x \to +0 $ and $ x \to +\infty $, respectively. Therefore the equation
		\[ 
		D + \frac{1}{1-x} + \frac{D+1}{x^{D+1} -1}
		= D \alpha 
		\]
		admits  a unique solution $ x_D(\alpha ) $ for every  fixed $ D>0 $ (not merely just integer $ D $)  and all $ \alpha \in (0, 1) $. Furthermore, the monotonicity of \eqref{eq:l2e00} implies that  $ x_D(0.5) = 1 $ and $ x_D (\alpha) < 1 $ when $ \alpha < 0.5 $.
		
		By $ F_\alpha(D) $ we denote $ \inf_{x > 0} \frac{P_D(x)}{x^{D \alpha}}  $;  we will show that
		\[ 
		\lim_{D \to +\infty} \frac{F_\alpha(D) }{D+1} = c_\alpha
		\]
		for some positive constant $ c_\alpha $. We will demonstrate this for $ \alpha < 0.5 $ (it is easy to see that $ F_\alpha(D) = F_{1-\alpha}(D) $); the proof consists of the following steps:
		\begin{itemize}
			\item We apply the substitution   $ x_D(\alpha) = \frac{D r}{D r +1} $ and express   $ F_\alpha(D)$ through $ r $, where $ r $ is the unique positive solution of   $ f(r,D)=0 $ for some $ f $.
			\item We demonstrate that $ r \to r^\infty $ for some finite $ r^\infty \geq 0 $ when $ D\to +\infty $.  This allows to find the limit $ \lim_{D \to +\infty} \frac{F_\alpha(D) }{D+1} = c_\alpha$, where $ c_\alpha  $ depends on $ r^\infty $ (which is completely determined by $ \alpha $). Since $ c_\alpha > 0 $, this establishes $ F_\alpha(D) = \Omega(D+1) $.
			\item To show that the limit $ r \to r^\infty $ exists, we shall show that $ r $ is decreasing in $ D $ by invoking the implicit function theorem. To that end, we show
			that  the partial derivatives of $ f(r,D) $ are positive, where $  f(r,D) =0 $ is the equation implicitly defining $ r $.
			
			\item We also show that $F_\alpha(D)/(D+1)$ is decreasing in $D$ on $[1,+\infty)$, therefore  $(D+1)c_\alpha$ is valid lower bound on $F_\alpha(D)$ for  all $D \geq 1$.   In order to achieve that, we will need to provide both lower and upper bounds on  $ x_D(\alpha) $, see \eqref{eq:l2_xBounds}.
		\end{itemize}

		Throughout the rest of this proof, we assume $ \alpha \in (0,  0.5 )$ to be fixed.
		\paragraph{Change of variables.}	 
		Let $ r =r_\alpha(D)> 0 $ satisfy $ x_D(\alpha) = \frac{D r}{D r +1}  $, i.e.,
		\[
		r = \frac{x_D(\alpha)}{D(1-x_D(\alpha))}.
		\]	  
		Substituting $ x_D(\alpha) = \frac{D r}{D r +1}  $ in  \eqref{eq:l2e00} and $ F_\alpha(D) $ yields
		\begin{equation}\label{eq:l2e01}
			F_\alpha(D)
			=
			\frac{1- (D r)^{D+1}  (D r +1)^{-D-1}   }{ 1-\frac{D r}{D r+ 1} }
			\cdot
			\frac{(D r+1)^{D \alpha}}{(D r)^{D \alpha }},
		\end{equation}
		where $ r $ is the unique\footnote{As witnessed by the fact that $x_D(\alpha)$ is unique and the mapping $ r\mapsto \frac{D r}{D r + 1} $ is bijective.} positive solution of 
		\begin{equation}\label{eq:l2e02}
			D  + (Dr +1)+ \frac{D+1}{ (D r)^{D+1}  (D r +1)^{-D-1}  -1 }=  D\alpha .
		\end{equation}
		
		We can transform \eqref{eq:l2e02} to 
		\[
		\frac{1}{ (D r)^{D+1}  (D r +1)^{-D-1}  -1 } +1
		=
		-\frac{D }{D+1} (r-\alpha) 
		\]	
		or, since $  1/ \left(x^{D+1} -1\right)  +1 =  1 / \left(1 -  x^{-D-1} \right)$,
		\[
		1 -  (D r)^{-D-1}  (D r +1)^{D+1}  = - \frac{D+1}{D (r-\alpha)} .
		\]
		The obtained equality can be rewritten as
		\begin{equation}\label{eq:l2e03}
			\frac{1}{r - \alpha}  = \frac{D}{D+1} \left(    \left(1 +  \frac{1}{D r}\right)^{D+1} -1    \right).
		\end{equation}
		Furthermore, the RHS of \eqref{eq:l2e01} can be simplified to 
		\begin{equation}\label{eq:l2e04}
			\frac{  \frac{D+1}{D+1 + D r - D\alpha }  }{1- \frac{D r}{D r+ 1} }
			\cdot
			\frac{(D r+1)^{D \alpha}}{(D r)^{D \alpha }}
			=
			(D+1) \,  \frac{Dr+1}{D+1 +D(r-\alpha)} \left( 1 + \frac{1}{D r} \right)^{D \alpha }.
		\end{equation}

		\paragraph{The large $ D $ limit.}	 
		We will show that the limit  $ \lim_{D \to +\infty} r_\alpha(D)   =: r^\infty \geq 0$ exists and is finite; taking limit of both sides of \eqref{eq:l2e03},  obtain
		\begin{equation}\label{eq:l2e03b}
			\frac{1}{r^\infty  - \alpha} = \mathrm e^{1 / r^\infty} - 1.
		\end{equation}
		Similarly we obtain from \eqref{eq:l2e04}
		\begin{equation}\label{eq:l2e04b}
		    	\lim_{D\to +\infty} \frac{F_\alpha(D)}{D+1} = \frac{r^\infty  \mathrm e^{ \alpha  / r^\infty}}{r^\infty +1- \alpha } 
		=
		r^\infty \mathrm e^{(\alpha-1)/r^\infty} \left(   \mathrm e^{1/r^\infty}-1 \right),
		\end{equation}
		i.e., 
		\[ 
		F_\alpha(D)  \sim    (D+1) c_\alpha,
		\]
		where   $ r^\infty $ is the unique positive solution  of \eqref{eq:l2e03b} and  $ c_\alpha :=   r^\infty \mathrm e^{(\alpha-1)/r^\infty} \left(   \mathrm e^{1/r^\infty}-1 \right)>0$ is independent of $ D $. The claim $ F_\alpha(D) = \Omega(D+1)  $ for $ \alpha = \Omega(1) $  follows. 

		In fact, it can be seen that $ c_\alpha $ is an increasing function in $ \alpha $ on $ (0,0.5) $. Notice that \eqref{eq:l2e03b} allows to express
		\begin{equation}\label{eq:l2e03c}
			\alpha = \frac{1}{1-\e^{r^{-\infty }}}+r^{\infty } ,
		\end{equation}
		and it is easy to verify that the RHS  is strictly increasing   in $ r^\infty $, approaching  values 0 and $ \frac{1}{2} $ as $ r^\infty $ approaches 0 and $+ \infty $. The inverse function theorem then implies that \eqref{eq:l2e03c} defines a differentiable, increasing function $ r^\infty(\alpha) $, defined on $ (0,0.5) $. Now let $ h(r,\alpha)$ stand for the RHS of \eqref{eq:l2e04b}, i.e.,  $ h(r,\alpha) =  r \mathrm e^{(\alpha-1)/r} \left(   \mathrm e^{1/r}-1 \right)$. Then we wish  to show  that $ h(r^\infty(\alpha), \alpha) $ is increasing in $ \alpha $, i.e., the total derivative $  \frac{\partial h}{\partial \alpha} (r,\alpha)  + \frac{\partial h}{\partial r} (r,\alpha)  \cdot \left( r^\infty \right)  ' (\alpha)  $ is positive when $ r = r^\infty(\alpha) $.
		
		However, the partial derivative $ \frac{\partial h}{\partial r} (r,\alpha) $ is zero when $ r = r^\infty(\alpha) $, since
		\[
		 \frac{\partial h}{\partial r} (r,\alpha)  = \frac{1}{r}\e^{\frac{\alpha -1}{r}} \left(\left(\e^{\frac{1}{r}}-1\right) (r-\alpha )-1\right).
		\]
		Moreover, for all $ r>0$ we  have
		\[
		\frac{\partial h}{\partial r} (r,\alpha)  =\left(\e^{\frac{1}{r}}-1\right) \e^{\frac{\alpha -1}{r}}>0,
		\]
		which allows to conclude  that $ c_\alpha =  h(r^\infty(\alpha), \alpha) $  is indeed increasing in $ \alpha $. Also, it can be calculated that $r^\infty(0.25) \approx 0.278279$ and $c_{0.25} \approx 0.664554$.

		\paragraph{The implicit function.}	 
		We proceed to verify that $ r_\alpha(D) $ converges as $ D \to +\infty $. To that end,  notice that \eqref{eq:l2e03} admits a unique positive solution for every $ D>0 $  (even though we were interested only in integer $ D $ before). Consequently, \eqref{eq:l2e03} defines an implicit function $ D \mapsto r_\alpha(D) $. We argue that $ r_\alpha(D) $ is decreasing in $ D $; since $ r_\alpha(D) > 0 $ is also lower-bounded by zero, this will establish that a nonnegative limit  $ r^\infty = \lim_{D \to +\infty} r_\alpha(D)   \in  (0,+\infty)$ exists (in fact, a positive limit, since it can be trivially shown  that $ r^\infty > \alpha $).
		
		Denote 
		\[
		f(r,D) := 
		\frac{D}{D+1} \left(    \left(1 +  \frac{1}{D r}\right)^{D+1} -1    \right) - \frac{1}{r - \alpha},
		\]
		then $ r= r_\alpha(D) $ is the positive solution of  $ f(r , D) = 0 $. We claim that  for every $ D_0>0 $ we have
		\[
		\left. \frac{\partial }{\partial r} f(r,D) \right \vert_{r = r_\alpha(D_0), D=D_0} > 0
		\quad\text{and}\quad
		\left. \frac{\partial }{\partial D} f(r,D) \right \vert_{r = r_\alpha(D_0), D=D_0} > 0.
		\]
		Then, by the implicit function theorem, in a neighborhood of $ D_0 $ a differentiable function  $ D \mapsto r_\alpha(D) $ satisfying $ f(r_\alpha(D),D)  =0$   exists and is  decreasing, since its derivative on that neighborhood is given by
		\begin{equation}\label{eq:l2decreasingR}
			r'(\alpha,D)
			=
			- \left.  \left( \frac{\partial }{\partial r} f(r,D)   \right)^{-1} \cdot    \frac{\partial }{\partial D} f(r,D) \right\vert_{r = r_\alpha(D)} < 0.
		\end{equation}
		The inequality $ r_\alpha(D+1)  < r_\alpha(D)$ then follows for positive $ D $, since the the argument can be repeated for every $ D_0 \in [D, D+1] $ and $r_\alpha(D_0)$ is unique.

		\paragraph{The partial derivatives of the implicit function.}
		It remains to consider the partial derivatives of $ f $. One finds that
		\begin{equation}\label{eq:l2der_f_D}
			\frac{\partial }{\partial D}f(r,D)
			=
			\frac{
				\left(\frac{1}{D r}+1\right)^D 
				\left((D+1) (D r+1) \ln \left(\frac{1}{D r}+1\right)-D-2   +r \right)
				-r
			}{(D+1)^2 r}.
		\end{equation}
		To establish
		$ \frac{\partial }{\partial D}f(r,D)>0 $,
		consider the mixed derivative
		\[
		\frac{\partial^2 }{\partial r \partial D}f(r,D)
		=
		\frac{\left(\frac{1}{D r}+1\right)^D \left(1-(D r+1) \ln \left(\frac{1}{D r}+1\right)\right)}{r^2 (D r+1)}<0,
		\]
		where the inequality can be obtained by setting $ z = 1/ (D r) $  in  $ \ln(1+z) \left(1 + 1/z\right)  > 1 $, which holds for all $ z>0 $.  Consequently, $ \frac{\partial }{\partial D}f(r,D) $ is decreasing in $ r $; since it is straightforward to check that $ \lim_{r\to +\infty} \frac{\partial }{\partial D}f(r,D) = 0 $, we conclude   $ \frac{\partial }{\partial D}f(r,D) > 0 $ for all $ r>0, D>0 $.

		Now consider
		\begin{equation}\label{eq:l2der_f_r}
			\frac{\partial }{\partial r}f(r,D)
			=
			\frac{1}{(r-\alpha )^2}-\frac{\left(\frac{1}{D r}+1\right)^D}{r^2}.
		\end{equation}
		We need to show that   $ \frac{\partial }{\partial r}f(r,D) > 0 $ when $ r  $ satisfies $ f(r,D)=0 $, i.e.,
		\[ 
		\frac{1}{r - \alpha}
		=	\frac{D}{D+1} \left(    \left(1 +  \frac{1}{D r}\right)^{D+1} -1    \right) .
		\]
		Notice that such $ r $ must satisfy $ r > \alpha $, since the RHS of the above equality is positive.
		
		Fix any $ r_0>0, D_0>0 $ such that $ f(r_0, D_0) = 0 $.
		Since  $ \frac{\partial f}{\partial r}(r_0,D) $ is decreasing in $ D $ (the mixed derivative is negative, as concluded previously) and
		\begin{equation}\label{eq:l2e05}
			\lim_{D \to +\infty} \frac{\partial f}{\partial r}(r_0,D)
			=
			\frac{1}{\left(r_0 - \alpha \right)^2}-\frac{\mathrm e^{\frac{1}{r_0}}}{r_0^2},
		\end{equation}
		there are two possibilities:
		\begin{enumerate}
			\item the RHS of \eqref{eq:l2e05} is nonnegative, then $  \frac{\partial f}{\partial r}(r_0,D) $ is positive for all $ D >0 $ including $ D=D_0 $ and we are done;
			\item $ r_0 >0$ is such that the RHS of \eqref{eq:l2e05} is negative, then there exists a unique $ D_1 > 0 $ satisfying $  \frac{\partial f}{\partial r} (r_0,D_1)  =0$.
		\end{enumerate}
		Consider the second possibility; to demonstrate   $ \frac{\partial f}{\partial r}(r_0,D_0)>0 $, we need to show that $ {D_1 > D_0} $, since   $ \frac{\partial f}{\partial r}(r_0,D) $ is decreasing in $ D $. This is equivalent to $ f(r_0, D_1) >f(r_0, D_0)=0$, since $ f(r_0,D) $ is strictly increasing in $ D $ as shown previously. 
		
		The equality $   \frac{\partial f}{\partial r} (r_0,D_1)  =0$ gives
		\[
		\frac{r_0^2}{(r_0-\alpha )^2}
		=
		\left(1+\frac{1}{D_1 r_0}\right)^{D_1},
		\]
		which allows to rewrite  $ f(r_0,D_1)  $ as 
		\[
		f(r_0,D_1) 
		=
		\frac{r_0 \left(2 \alpha  D_1+D_1+2\right)-\alpha  \left(\alpha  D_1+D_1+1\right)}{\left(D_1+1\right) \left(r_0 - \alpha \right)^2}
		\]
		Now recall that $ r_0 > \alpha $, therefore the numerator of the above expression can be lower-bounded as 
		\begin{align*}
			& r_0 
			\left(2 \alpha  D_1+D_1+2\right)
			-\alpha  
			\left(\alpha  D_1+D_1+1\right)
			\\
			>
			& 
			\alpha 
			\left(
			\left(2 \alpha  D_1+D_1+2\right)
			-
			\left(\alpha  D_1+D_1+1\right)
			\right)
			=
			\alpha (\alpha D_1 + 1)	 > 0.
		\end{align*}
		Consequently, $ f(r_0,D_1) >0 $, implying that $ D_1 > D_0 $ and $   \frac{\partial f}{\partial r}(r_0,D_0) > 0$. This establishes \eqref{eq:l2decreasingR} and the existence of $ r^\infty $.
		
		\paragraph{Monotonicity of the lower bound in \texorpdfstring{$\boldsymbol{D}$}{D}.}
		It remains to show that $ \frac{F_\alpha(D)}{D+1}  $ is decreasing in $ D $, therefore $ (D+1)c_\alpha  $ is a valid lower bound on $ F_\alpha(D) $ for all $ D $. From \eqref{eq:l2e04} it follows that $ g(r_\alpha(D),D) $ must be decreasing, where
		\[
		g(r,D):=
		 \frac{Dr+1}{D+1 +D(r-\alpha)} \left( 1 + \frac{1}{D r} \right)^{D \alpha }.
		\]
		Since $r _\alpha$ is  differentiable, it suffices to show that the total derivative of $ g(r_\alpha(D),D) $ is non-positive, i.e.,
		\[
		\frac{\partial }{\partial D} g(r,D) + \frac{\partial }{\partial r} g(r,D) \cdot r'_\alpha(D) \leq  0.
		\]
		Taking into account \eqref{eq:l2decreasingR} and the now-proven fact that the partial derivatives of $f$ are positive (at the values of $ (r,D) $ we are interested in),  we can equivalently transform the above inequality as
		\begin{equation}\label{eq:l2e06}
			\frac{\partial }{\partial D} g(r,D) \cdot   \frac{\partial }{\partial r} f(r,D) - \frac{\partial }{\partial r} g(r,D) \cdot   \frac{\partial }{\partial D}f (r,D) \leq   0,
		\end{equation}
		which must be satisfied with  $ r = r_\alpha(D) $ for all $ D \geq 1 $.
		
		The derivatives $ \frac{\partial }{\partial r} f(r,D)   $  and $ \frac{\partial }{\partial D} f(r,D)  $ are given by \eqref{eq:l2der_f_r} and \eqref{eq:l2der_f_D}, respectively; however, the equality \eqref{eq:l2e03} allows to express them (when $ r= r_\alpha(D) $) as
		\begin{align*}
			& 	\frac{\partial f}{\partial D}(r,D)
			=
			\frac{(D (-\alpha +r+1)+1) \left((D+1) (D r+1) \ln \left(\frac{1}{D r}+1\right)-D+r-2\right)}{(D+1)^2 (D r+1) (r-\alpha )}-\frac{1}{(D+1)^2}, 
			\\
			& 	\frac{\partial f}{\partial r}(r,D)
			=
			\frac{\alpha +D \left(-\alpha ^2+\alpha +2 \alpha  r-r\right)}{r (D r+1) (r-\alpha )^2},
		\end{align*}
	valid when $ D>0 $, $ r= r_\alpha(D) $.
	
	While the expressions of  $ \frac{\partial g}{\partial D}  $, $ \frac{\partial g}{\partial r}  $ are somewhat cumbersome,
	\begin{align*}
		&  \frac{\partial g}{\partial D} (r,D)=
		\frac{\left(\frac{1}{D r}+1\right)^{\alpha  D} \left(\alpha  D (\alpha -r-1)+\alpha  (D r+1) (D (-\alpha +r+1)+1) \ln \left(\frac{1}{D r}+1\right)-1\right)}{(D (-\alpha +r+1)+1)^2},
		\\
		&  \frac{\partial g}{\partial r} (r,D) =
		-\frac{D \left(\frac{1}{D r}+1\right)^{\alpha  D} \left(\alpha +D \left(-\alpha ^2+\alpha +2 \alpha  r-r\right)\right)}{r (D (-\alpha +r+1)+1)^2} ,
	\end{align*}
now  \eqref{eq:l2e06} 	can be rewritten in a rather simple form:
\[
\frac{\left(\frac{1}{D r}+1\right)^{\alpha  D} \left(\alpha +D \left(-\alpha ^2+\alpha +2 \alpha  r-r\right)\right) \left((\alpha +D r) \ln  \left(\frac{1}{D r}+1\right)-1\right)}{(D+1) r (r-\alpha )^2 (D (-\alpha +r+1)+1)}
\leq 
0.
\]
Since the denominator and $ \left(\frac{1}{D r}+1\right)^{\alpha  D} $ are obviously positive, it suffices to demonstrate that 
\begin{align*}
	&  \alpha +D \left(-\alpha ^2+\alpha +2 \alpha  r-r\right) \geq 0, 
	\\
	& (\alpha +D r) \ln  \left(\frac{1}{D r}+1\right)-1 \leq 0.
\end{align*}
At this point it is more advantageous to return to the variable $ x = \frac{D r }{D r + 1} $; in $ x $ both inequalities become (after multiplying by $ 1-x>0 $)
\begin{align}
	& \alpha +D \alpha (1-\alpha ) - (1-\alpha ) x (\alpha  D+1) \geq 0, \label{eq:l2e07a}  \\
	& -(\alpha + x(1-\alpha  )) \ln x - (1-x)  \leq 0.\label{eq:l2e07b}
\end{align}
In the following paragraph we will show that $ x = x_D(\alpha) $   satisfies bounds
\begin{equation}\label{eq:l2_xBounds}
	\frac{\alpha }{1-\alpha } \leq x_D(\alpha)  \leq  \frac{\alpha }{1-\alpha } \cdot    \frac{ 1+ (1-\alpha ) D  }{ 1+ \alpha  D } ,
	\quad  D \geq 1.
\end{equation}
Notice that \eqref{eq:l2e07a} is equivalent to the upper bound on $ x_D(\alpha) $ in \eqref{eq:l2_xBounds}; let us show that \eqref{eq:l2e07b} holds. To that end, differentiate the LHS of \eqref{eq:l2e07b} twice w.r.t. $ x $ to see that the second derivative $  - \frac{(1-\alpha ) x-\alpha }{x^2} $ is non-positive due to the lower bound on $ x_D(\alpha) $ in \eqref{eq:l2_xBounds}. Moreover, both the LHS of \eqref{eq:l2e07b}  and its derivative equal zero at $ x=1 $. We conclude that \eqref{eq:l2e07b} is satisfied by all $ \frac{\alpha }{1-\alpha }  \leq x \leq  1  $, therefore  also by $ x = x_D(\alpha) $. This completes the argument that the total derivative of $ g(r_\alpha(D),D) $ is non-positive.

\paragraph{Lower and upper bound on \texorpdfstring{$\boldsymbol{x_D}$}{x_D}.}
The final part is to show that for all $D \geq 1$ and $  \alpha \in (0,0.5) $ the value of $ x_D(\alpha)  $ satisfies  \eqref{eq:l2_xBounds}. Instead of trying to prove directly bounds on the implicitly defined $ x_D(\alpha)  $, we can exploit the monotonicity of  \eqref{eq:l2e00} (whose RHS is equal to $ D \alpha  $) and argue that the following inequalities hold for all $ x \in (0,1)  $ and $ D \geq 1 $:
\[
x - \frac{\alpha(x) }{1-\alpha(x) } \geq 0 ,   \quad   \frac{\alpha(x) }{1-\alpha (x)} \cdot    \frac{ 1+ (1-\alpha(x) ) D  }{ 1+ \alpha(x)  D }  -x \geq 0,
\]
where 
\[
 \alpha(x) := 
 1 + \frac{1}{D(1-x)} + \frac{D+1}{D(x^{D+1} -1)}, \quad  x\in (0,1).
\]
After simplification, both inequalities in 	question are equivalent to
\begin{align*}
	& 
	 \frac{x \left( 
%
	 	(D-1)  x^{D+1}
	 	-(D+1) x^D
	 	+(D+1) x-(D-1)
	 	 \right)}{x^{D+1}-(D+1) x+D}
	\leq 0, 
	\\
	& 
	 \frac{
	 	(1-x) x 
	 	\left(
	 	x^{2 D+2}
	 	-(D+1)^2 x^{D+2}
	 	+2 D (D+2) x^{D+1}
	 	-(D+1)^2 x^D
	 	+1
	 	\right)
 	}{
 	\left( 
 	(D+1) x^{D+2}-(D+2) x^{D+1}+1
 	\right) 
 	\left(
 	x^{D+1}-(D+1) x+D
 	\right)
 }
	 	\geq 0
\end{align*}
or
\begin{equation}\label{eq:l2e08}
	\frac{x p_4(x)}{p_1(x)} \geq  0
	\quad\text{and}\quad
	\frac{x(1-x) p_3(x)}{p_1(x) p_2(x)} \geq 0,
\end{equation}
where we denote
\begin{align*}
	& p_1(x) := x^{D+1}-(D+1) x+D   \\
	& p_2(x):= 	(D+1) x^{D+2}-(D+2) x^{D+1}+1  \\
	&p_3(x):= x^{2 D+2} 	-(D+1)^2 x^{D+2} 	+2 D (D+2) x^{D+1} 	-(D+1)^2 x^D +1   \\
	&p_4(x):= -(D-1)  x^{D+1} +(D+1) x^D -(D+1) x+(D-1) .
\end{align*}
Since $ x \geq 0 $, $ 1-x \geq 0 $, it suffices to show that $ p_j(x) >0  $, $ j\in \{1,2,3,4\} $, for all $D \geq 1  $ and $ x \in (0,1)$ (with the exception $ p_4(x) \equiv 0 $ when $ D=1 $). This can be done by observing that all $ p_j $   satisfy $ p_j(1)  =0$ and proving that $ p_j $ are decreasing on $ (0,1) $ (or non-increasing, in the case of $ p_4 $ and $ D=1 $). 

To show the monotonicity of $ p_j $, consider their derivatives $ p_j' $. We easily obtain $ p_1' (x) = (D+1) \left(x^D-1\right) < 0$ and $ p_2'(x) = (D+1) (D+2) (x-1) x^D  <0 $ on $ (0,1) $. It remains to show $ p_3'(x) < 0 $ and $ p_4'(x) < 0 $. Since 
\[
 p_3'(x) =  -(D+1) x^{D-1} \left(
-2 x^{D+2}+ x^2 \left(D^2+3 D+2\right)-2 D (D+2) x+D (D+1)
 \right) ,
\]
we need to show that $ q(x) := -2 x^{D+2}+ x^2 \left(D^2+3 D+2\right)-2 D (D+2) x+D (D+1)  $ is positive on $ (0,1) $. However, $ q(1) = 0 $ and 
\[
q'(x) =-2(D+2) p_1(x)  < 0,\  x \in (0,1),
\]
therefore $ q $ is decreasing and positive on $ (0,1 ) $, and so is $ p_3 $.

Finally, consider
\[
p_4'(x) = -(D+1) (1-x) \left(
\frac{1-x^D}{1-x}-D x^{D-1}
\right).
\]
We see that $ p_4 $  is non-increasing on $ (0,1) $ iff 
\begin{equation}\label{eq:l2e09}
	\frac{1-x^D}{1-x} 
	\geq 
	D x^{D-1}, \quad  x\in (0,1), \,  D \geq 1.
\end{equation}
However, $ \frac{1-x^D}{1-x}  $ is the divided difference of the function $ \phi(t) = t^{D} $ on the nodes $ t_0 =1 $, $ t_1=x $, whereas $ D x^{D-1} $ is the derivative of $ \phi $ at $ x $.
By the mean value theorem for divided differences, the value of the divided difference $ \frac{1-x^D}{1-x}  $ is equal to  to $ \phi'(t_0) $ for some $ t_0  $ strictly between $ x $ and 1, and \eqref{eq:l2e09} is equivalent to $ \phi'(t_0)  \geq  \phi(x) $. However, $ \phi'(t) = D  t^{D-1} $ is non-decreasing (even increasing when $ D>1 $), therefore $ t_0 > x $ implies the desired $  \phi'(t_0)  \geq  \phi(x)   $  (and the inequality is strict for $ D>1 $). This establishes   \eqref{eq:l2_xBounds} and completes the proof of the lemma.
\end{proof}

\subsection{Quantum search lower bound}

\begin{lemma} \label{thm:slb}
If $\alpha_{1,D} \leq \frac{1}{4}$, then the query complexity of the quantum search part of the algorithm is $\widetilde\Omega\left(\left(\frac{D+1}{\e}\right)^n\right)$.
\end{lemma}

\begin{proof}
We will prove the lower bound in the following way.
Suppose that $\alpha_{k,D} \leq \frac{1}{4}$ for some $k < K+1$ and $\alpha_{i,D} > \frac{1}{4}$ for all $k < i \leq K+1$ (recall that by our convention $\alpha_{K+1,D}=\frac{1}{2}$).
Then we will prove that the complexity of running the VTS over the layers $\mathcal L_{K+1}$, $\mathcal L_{K}$, $\ldots$, $\mathcal L_{k}$ successively and then running the \textsc{Path} recursively between layers $\mathcal L_k$ and $\mathcal L_{k+1}$ requires $\widetilde\Omega\left(\left(\frac{D+1}{\e}\right)^n\right)$ queries altogether.

More specifically, examine the process of the quantum search over the mentioned layers.
During step \ref{itm:sc3}, the VTS examines vertices $v^{(K+1)} \in \mathcal L_{K+1}$.
At the layer $\mathcal L_i$, the VTS in the \textsc{LayerPath} procedure examines vertices $v^{(i)} \in \mathcal L_i$ such that $v^{(i)} < v^{(i+1)}$.

Now, for each $i \in [k,K+1]$ we will examine only specific types of vertices $v^{(i)} \in \mathcal L_i$.
We will have the property that for each $j \in [k,K]$ and for each examined $v^{(j+1)}$, we examine at least $N_j$ vertices $v^{(j)}$.
We also define $N_{K+1}$ simply as the number of examined vertices $v^{(K+1)}$ of specific type.
Then the successive VTS calls of \textsc{LayerPath} will examine $N_k \cdots N_{K+1}$ sequences of vertices $(v^{(k)},\ldots,v^{(K+1)})$.
If the query complexity of $\textsc{Path}(v^{(k)},v^{(k+1)})$ is $T$, then the total query complexity is at least
$$\widetilde\Omega(T) \cdot \prod_{i=k}^{K+1} \widetilde\Omega(\sqrt{N_i}) =\widetilde\Omega(T\sqrt{N_k \cdots N_{K+1}}).$$
Note that here we require that both $T$ and $N_i$ are exponential in $n$ because of $\widetilde\Omega$ notation, which will be apparent in the proof later.
Also note that the product of at most $K+2$ expressions $\widetilde\Omega$ above is still $\widetilde\Omega$, as $K$ is fixed.

Moreover, we can lower bound $T$ by examining the number of vertices in the middle layer of the sublattice bounded by the vertices $v^{(k)}$ and $v^{(k+1)}$.
Let the number of vertices in this sublattice be $S$; then the middle layer of this sublattice has size at least $\frac{S}{Dn} = \widetilde\Omega(S)$, as the middle layer has the largest size (and, as we will see, $S$ is also exponential in $n$).
Since in the call of $\textsc{Path}(v^{(k)},v^{(k+1)})$ the first VTS examines all vertices in the middle layer, we have $T = \widetilde\Omega(\sqrt S)$.

Therefore, the task now is reduced to showing that we can find such types of vertices $v^{(i)}$ so that
$$S \cdot N_k \cdots N_{K+1} = \widetilde\Omega\left(\left(\frac{D+1}{\e}\right)^{2n}\right).$$

First, we prove the following lemma.

\begin{lemma} \label{thm:struct}
Examine the vertices $v \in \{0,1,\ldots,D\}^n$ such that for each $d \in \{0,1,\ldots,D\}$ we have $|\{i \mid v_i = d\}| = \frac{n}{D+1}$.
Then $|v| = \frac{nD}{2}$ and the number of such vertices is $\widetilde\Omega((D+1)^n)$.
\end{lemma}

\begin{proof}
First we can see that
$$|v| = \sum_{d=0}^D d\cdot \frac{n}{D+1} = \frac{n}{D+1} \cdot \frac{D(D+1)}{2} = \frac{nD}{2}.$$
The number of such vertices on the other hand is given by the multinomial coefficient
\begin{equation} \label{eq:mnm}
    \binom{n}{\frac{n}{D+1},\ldots,\frac{n}{D+1}} = \frac{n!}{\left(\frac{n}{D+1}!\right)^{D+1}}.
\end{equation}
Using standard bounds for the factorial, $\sqrt{2 \pi} n^{n+\frac{1}{2}} \e^{-n} \leq n! \leq \e n^{n+\frac{1}{2}} \e^{-n}$, we get that (\ref{eq:mnm}) is at least
\begin{align*}
    \frac{\sqrt{2 \pi} n^{n+\frac{1}{2}} \e^{-n}}{\left( \e \left( \frac{n}{D+1}\right)^{\frac{n}{D+1}+\frac{1}{2}} \e^{-\frac{n}{D+1}}\right)^{D+1} }
    &=  \frac{\sqrt{2 \pi n}}{\e^{D+1}} \frac{n^n}{\left( \frac{n}{D+1}\right)^{n+\frac{D+1}{2}}} \\
    &=  \frac{\sqrt{2 \pi n}}{\e^{D+1}} \left(\frac{D+1}{n}\right)^{\frac{D+1}{2}} (D+1)^n \\
    &= \widetilde\Omega((D+1)^n). \qedhere
\end{align*}
\end{proof}

Now we give the description of the vertices $v^{(i)}$.
\begin{itemize}
\item For $\mathcal L_{K+1}$, we take the vertices $v^{(K+1)}$ as is described by Lemma \ref{thm:struct}.
The number of such vertices is $$N_{K+1} = \widetilde\Omega((D+1)^n).$$
For a fixed such vertex $v^{(K+1)}$, define $I_d = \{i \mid v^{(K+1)}_i = d\}$.
\item Now let's define vertices $v^{(i)} \in \mathcal L_i$ for $k < i < K+1$.
Let
$$\gamma_i \coloneqq 2-4\alpha_{i,D}.$$
It is helpful to think of $\gamma_i$ as the coefficient telling the distance of $\alpha_{i,D}$ from $\frac{1}{2}$: if $\alpha_{i,D} = \frac{1}{2}$, then $\gamma_i = 0$, and if $\alpha_{i,D}=\frac{1}{4}$, then $\gamma_i = 1$.
We examine vertices $v^{(i)}$ with the requirement that for each $d \in [0,D]$, we can partition $I_d = A_d \cup B_d$ so that
\begin{enumerate}[(a)]
    \item $\forall j \in A_d : v^{(i)}_j = d$ and $|A_d| = (1-\gamma_i)\cdot \frac{n}{D+1}$;
    \item $\sum_{j \in B_d} v^{(i)}_j = \frac{1}{2} \cdot d\cdot \gamma_i \cdot \frac{n}{D+1}$ and $|B_d| = \gamma_i \cdot \frac{n}{D+1}$.
\end{enumerate}
Note that the above conditions also hold if $i=K+1$ (we can assume that $\gamma_{K+1} = 0$).
First we can make sure that the weight of such vertex is equal to $\alpha_{i,D} \cdot Dn$:
\begin{align*}
   |v^{(i)}| &=
   \sum_{d=0}^D \left( d\cdot (1-\gamma_i)\cdot \frac{n}{D+1} + \frac{1}{2} \cdot d\cdot \gamma_i \cdot \frac{n}{D+1}\right) \\
   &= \left(1-\frac{\gamma_i}{2}\right) \cdot \frac{n}{D+1} \cdot \sum_{d=0}^D d \\
   &= \left(1-\frac{\gamma_i}{2}\right) \cdot \frac{n}{D+1} \cdot \frac{D(D+1)}{2} \\
   &= \left(\frac{1}{2}-\frac{\gamma_i}{4}\right) \cdot Dn = \alpha_{i,D} \cdot Dn.
\end{align*}

Now take any such vertex $v^{(i+1)}$ for $i \in [k+1,K]$, with its sets $A_d$, $B_d$.
We can construct the vertices $v^{(i)} < v^{(i+1)}$ that satisfy the same conditions as follows.
For each $d \in [0,D]$:
\begin{enumerate}
    \item Take some fixed $A_d' \subset A_d$ such that $|A_d'| = (1-\gamma_{i})\cdot \frac{n}{D+1}$.
    This is possible, since $\gamma_{i} > \gamma_{i+1}$ (as $\alpha_{i} < \alpha_{i+1}$).
    For all $j \in A_d'$, let $v^{(i)} = d$.
    \item For all $j \in B_d$, let $v^{(i)}_j = v^{(i+1)}_j$.
    \item Let the set of other coordinates be $C \coloneqq A_d \setminus A_d'$.
    For $j \in C$, pick $v^{(i)}_j$ such that
    $$\sum_{j \in C} v^{(i)}_j = \frac{1}{2}\cdot d \cdot |C|.$$
\end{enumerate}
We can see that such $v^{(i)}$ satisfies the proposed conditions with $A_d'$ and $B_d' = C \cup B_d$.
By Lemma \ref{thm:struct}, the number of choices for the values of $v^{(i)}_j$ in the positions $C$ is $$\widetilde\Omega\left((d+1)^{|C|}\right) = \widetilde\Omega\left((d+1)^{(\gamma_i-\gamma_{i+1})\cdot \frac{n}{D+1}}\right).$$
Combining together these numbers for all possible $d$, we get that
$$N_i = \prod_{d=0}^D \widetilde\Omega\left((d+1)^{(\gamma_i-\gamma_{i+1})\cdot \frac{n}{D+1}}\right).$$
Since $D$ is fixed, we take the $\widetilde\Omega$ outside the product,
$$N_i = \widetilde\Omega\left(\left(\prod_{d=0}^D(d+1)\right)^{(\gamma_i-\gamma_{i+1})\cdot \frac{n}{D+1}}\right) = \widetilde\Omega\left(\left((D+1)!\right)^{(\gamma_i-\gamma_{i+1})\cdot \frac{n}{D+1}}\right).$$
Now we can use the well-known lower bound for the factorial, $n! \geq \left(\frac{n}{\text{e}}\right)^n$.
Then we get
$$N_i = \widetilde\Omega\left(\left(\left(\frac{D+1}{\text{e}}\right)^{D+1}\right)^{(\gamma_i-\gamma_{i+1})\cdot \frac{n}{D+1}}\right) = \widetilde\Omega\left(\left(\frac{D+1}{\text{e}}\right)^{(\gamma_i-\gamma_{i+1})\cdot n}\right).$$
\item Finally, we define vertices $v^{(k)} \in \mathcal L_k$.
Now let
$$\gamma_k \coloneqq 4\alpha_{k,D}.$$
In this case, $\alpha_{k,D} \leq \frac{1}{4}$.
Here $\gamma_k$ describes the distance from $\alpha_{k,D}$ to $0$: if $\alpha_{k,D} = 0$, we have $\gamma_k = 0$ and if $\alpha_{k,D} = \frac{1}{4}$, we have $\gamma_k = 1$.
Again, take a vertex $v^{(k+1)}$ satisfying the earlier conditions, with its sets $A_d$, $B_d$.
Now we distinguish two cases:
\begin{enumerate}[(a)]
    \item $\gamma_k \leq \gamma_{k+1}$.
    
    Construct $v^{(k)}$ as follows:
    \begin{enumerate}[1.]
        \item For all $j \in A_d$, let $v^{(k)}_j = 0$.
        \item For $j \in B_d$, pick any values for $v^{(k)}_j$ such that $$\sum_{j \in B_d} v^{(k)}_j = \frac{1}{2}\cdot d \cdot \gamma_k \cdot \frac{n}{D+1}.$$
        We can check that it is possible to assign such values to have $v^{(k)} < v^{(k+1)}$, since $$\sum_{j \in B_d} v^{(k+1)}_j = \frac{1}{2}\cdot d \cdot \gamma_{k+1} \cdot \frac{n}{D+1} \geq \frac{1}{2}\cdot d \cdot \gamma_{k} \cdot \frac{n}{D+1}.$$
    \end{enumerate}
    The weight of $v^{(k)}$ is equal to
    \begin{align*}
        |v^{(k)}| 
        = \sum_{d=0}^D \frac{1}{2}\cdot d \cdot \gamma_{k} \cdot \frac{n}{D+1}
        = \frac{1}{2} \cdot \gamma_k \cdot \frac{n}{D+1} \cdot \frac{D(D+1)}{2} = \alpha_{k,D} \cdot Dn.
    \end{align*}
    
    Now we can calculate the values of $N_k$ and $S$.
    We examine only a single choice of the values $v_j^{(k)}$ for $j \in B_d$, hence $$N_k = 1.$$
    On the other hand, for each $d \in [0,D]$, there are $|A_d| = (1-\gamma_{k+1})\cdot \frac{n}{D+1}$ positions such that $v^{(k)}_j = 0$ and $v^{(k+1)}_j = d$.
    Therefore, $$S = \prod_{d=0}^D (d+1)^{(1-\gamma_{k+1}) \cdot \frac{n}{D+1}} = \left(((D+1)!)^{\frac{1}{D+1}}\right)^{(1-\gamma_{k+1})\cdot n} \geq \left(\frac{D+1}{\text{e}}\right)^{(1-\gamma_{k+1})\cdot n}.$$
    
    \item $\gamma_k > \gamma_{k+1}$.
    
    Construct $v^{(k)}$ as follows:
    \begin{enumerate}[1.]
        \item Pick any $A_d' \subset A_d$ such that $|A_d'| = (1-\gamma_k) \cdot \frac{n}{D+1}$.
        For all $j \in A_d'$, let $v^{(k)}_j = 0$.
        \item For all $j\in B_d$, let $v^{(k)}_j = v^{(k+1)}_j$.
        \item Let $C \coloneqq A_d \setminus A_d'$.
        For $j \in C$, pick any values for $v^{(k)}_j$ such that $$\sum_{j \in C} v^{(k)}_j = \frac{1}{2}\cdot d \cdot |C|.$$
    \end{enumerate}
    The weight of $v^{(k)}$ again is equal to
    \begin{align*}
        |v^{(k)}| &= \sum_{d=0}^D \frac{1}{2} \cdot d \cdot (|C|+|B_d|) = \sum_{d=0}^D \frac{1}{2} \cdot d \cdot \left(\frac{n}{D+1}-|A_d'|\right) \\ &= \sum_{d=0}^D \frac{1}{2} \cdot d \cdot \gamma_k \cdot \frac{n}{D+1} = \alpha_{k,D} \cdot Dn.
    \end{align*}
\end{enumerate}
Now we calculate the values of $N_k$ and $S$.
Since $|C| = (\gamma_k - \gamma_{k+1}) \cdot \frac{n}{D+1}$, the number of choices for the values of $v^{(k)}_j$ for $j \in C$ is $\widetilde\Omega\left((d+1)^{(\gamma_k - \gamma_{k+1}) \cdot \frac{n}{D+1}}\right)$ by Lemma \ref{thm:struct}.
Taking into account all $d \in [0,D]$, we get that
\begin{align*}
    N_k 
    &= \prod_{d=0}^D \widetilde\Omega\left((d+1)^{(\gamma_k - \gamma_{k+1}) \cdot \frac{n}{D+1}}\right) \\
    &= \widetilde\Omega\left( ((D+1)!)^{(\gamma_k - \gamma_{k+1}) \cdot \frac{n}{D+1}}\right) \\
    &= \widetilde\Omega\left( \left(\frac{D+1}{\text{e}}\right)^{(\gamma_k - \gamma_{k+1}) \cdot n}\right).
\end{align*}
On the other hand, for each $d \in [0,D]$, there are $|A_d'| = (1-\gamma_k)\cdot \frac{n}{D+1}$ positions such that $v^{(k)}_j = 0$ and $v^{(k+1)}_j = d$.
    Therefore, $$S = \prod_{d=0}^D (d+1)^{(1-\gamma_k) \cdot \frac{n}{D+1}} = \left(((D+1)!)^{\frac{1}{D+1}}\right)^{(1-\gamma_k)\cdot n} \geq \left(\frac{D+1}{\text{e}}\right)^{(1-\gamma_k)\cdot n}.$$
\end{itemize}

In both cases, we can see that
$$N_k \cdot S = \widetilde\Omega\left(\left(\frac{D+1}{\text{e}}\right)^{(1-\gamma_{k+1})\cdot n}\right).$$
Finally, taking into account all other values of $N_i$, we obtain
\begin{align*}S\cdot \prod_{i=k}^{K+1} N_i
&= (S \cdot N_k) \cdot \left(\prod_{i=k+1}^K N_i\right) \cdot N_{K+1} \\ &= \widetilde\Omega\left(\left(\frac{D+1}{\text{e}}\right)^{(1-\gamma_{k+1})\cdot n}\right) \cdot \left(\prod_{i=k+1}^{K+1} \widetilde\Omega\left(\left(\frac{D+1}{\text{e}}\right)^{(\gamma_i-\gamma_{i+1})\cdot n}\right)\right) \cdot \widetilde\Omega\left(\left(\frac{D+1}{\text{e}}\right)^n\right) \\
&= \widetilde\Omega\left(\left(\frac{D+1}{\text{e}}\right)^{2n}\right)
\end{align*}
Therefore, the total complexity of the quantum search is lower bounded by
\begin{equation*}
    \widetilde\Omega\left( \left(\frac{D+1}{\text{e}}\right)^n \right). \qedhere
\end{equation*}
\end{proof}

\end{document}